\documentclass[letterpaper,twocolumn,10pt]{article}
\usepackage{usenix}

\usepackage{tikz}

\usepackage{filecontents}
\usepackage{amsmath,amssymb,amsfonts,amsthm}
\usepackage{algorithm}
\usepackage{algorithmic}
\usepackage{graphicx}
\usepackage{textcomp}
\usepackage{xcolor}
\usepackage{hyperref}
\usepackage{xspace}
\usepackage{multirow}
\usepackage{caption}
\usepackage{subcaption}
\usepackage{soul}
\usepackage{thmtools}

\floatname{algorithm}{\textbf{Protocol}}

\theoremstyle{definition}
\newtheorem{definition}{Definition}
\theoremstyle{plain}
\newtheorem{lemma}{Lemma}
\declaretheorem[name=Theorem]{theorem}
\newtheorem{corollary}{Corollary}
\newtheorem{assumption}{Assumption}

\usepackage{todonotes}

\newcommand{\vc}{\textsf{Thunderdome}\xspace}

\soulregister{\cite}7
\soulregister{\citep}7
\soulregister{\citet}7
\soulregister{~\ref}7
\soulregister{\pageref}7
\soulregister{\vc}7

\begin{document}
\title{\Large \bf\vc: Timelock-Free Rationally-Secure Virtual Channels \\
}

\author{
    Zeta Avarikioti \\
    \textit{TU Wien \& Common Prefix} 
    \and
    Yuheng Wang \\
    \textit{TU Wien} 
    \and
    Yuyi Wang \\
    \textit{CRRC Zhuzhou Institute \& Tengen Intelligence Institute} 
}



\maketitle

\begin{abstract}

Payment channel networks (PCNs) offer a promising solution to address the limited transaction throughput of deployed blockchains. However, several attacks have recently been proposed that stress the vulnerability of PCNs to timelock and censoring attacks. To address such attacks, we introduce \vc, the first timelock-free PCN. Instead, \vc leverages the design rationale of virtual channels to extend a timelock-free payment channel primitive, thereby enabling multi-hop transactions without timelocks. 
Previous works either utilize timelocks or do not accommodate transactions between parties that do not share a channel.



At its core, \vc relies on a committee of \emph{non-trusted} watchtowers, known as wardens, who ensure that no honest party loses funds, even when offline, during the channel closure process. We introduce tailored incentive mechanisms to ensure that all participants follow the protocol's correct execution.  Besides a traditional security proof that assumes an honest majority of the committee, we conduct a formal game-theoretic analysis to demonstrate the security of \vc when all participants, including wardens, act rationally. We implement a proof of concept of \vc on Ethereum to validate its feasibility and evaluate its costs. Our evaluation shows that deploying \vc, including opening the underlying payment channel, costs approximately \$15 (0.0089 ETH), while the worst-case cost for closing a channel is about \$7 (0.004 ETH).
\end{abstract}

\section{Introduction}

Blockchains have introduced a groundbreaking decentralized financial paradigm through cryptocurrencies, eliminating the need for trusted intermediaries~\cite{zheng2018blockchain}. However, blockchains face major scalability limitations, primarily due to the requirement that each transaction be validated by all nodes in the network~\cite{croman2016scaling}. As a result, popular cryptocurrencies like Bitcoin~\cite{nakamoto2008bitcoin} and Ethereum suffer from relatively slower processing speeds and much lower transaction throughput compared to centralized payment systems like Visa, hindering their potential for mass adoption. While solutions such as more efficient consensus protocols~\cite{blackshear2023sui,spiegelman2022bullshark,danezis2022narwhal} can increase throughput, they typically require changes to existing protocols, often causing a hard fork. To overcome these scalability challenges, Payment Channel Networks (PCNs) have emerged as a promising alternative~\cite{gudgeon2020sok}.

A payment channel functions as a secure off-chain "joint account" between two parties, utilizing the blockchain only for the opening and closing phases. This approach substantially reduces the number of on-chain transactions, effectively addressing the scalability issues of the underlying blockchain. Payment channels involve three primary operations: \textit{Open}, \textit{Update}, and \textit{Close}. The channel is opened by locking a certain amount of coins on-chain, referred to as the channel balance. The parties can then transact off-chain by updating their channel balance through the exchange of digitally signed messages that specify the new distribution of coins. Finally, the channel can be closed, settling all off-chain transactions with the final agreed-upon state posted on-chain. To prevent fraud, payment channel protocols often include a dispute period with a \emph{timelock}, allowing a counterparty to penalize a potentially dishonest party. For example, in the Bitcoin Lightning Network~\cite{poon2016bitcoin}, a party can claim the entire channel balance if the other party posts an outdated update.

An overlay Payment Channel Network (PCN) operates as a Layer 2 solution on top of a single blockchain (Layer 1), primarily functioning off-chain to address the scalability limitations of the base network. Notably, a PCN allows parties with at least one direct payment channel to make payments across the network to other parties, even if they do not share a direct channel. To facilitate multi-hop payments from sender to receiver, a path of sufficiently funded channels is identified. All channels along this path are then updated using lock contracts with a \emph{timelock} to ensure the atomic execution of the payment. Techniques such as Hash Timelock Contracts (HTLCs), used in the Bitcoin Lightning Network~\cite{poon2016bitcoin}, and methods like adaptor signatures~\cite{tairi2021post,aumayr2021generalized} or Verifiable Timed Signatures (VTS)~\cite{thyagarajan2020verifiable}, are commonly employed in PCNs. These techniques rely on timelocks to ensure that parties can recover their coins in the event of a payment failure. Thus, \emph{timelocks are essential for the security of both payment channel primitives and multi-hop payment protocols in PCNs.}

Recent research reveals that timelocks can introduce new attack vectors, notably censorship attacks such as timelock bribing. These attacks show that PCNs relying on timelock contracts for multi-hop payments might allow an adversary to steal the entire channel balance by delaying the on-chain inclusion of a specific transaction~\cite{nadahalli2021timelocked, tsabary2021mad, chung2022rapidash, garewal2020helium}. Similar vulnerabilities are present in the payment channel primitives, as noted in~\cite{avarikioti2019brick, avarikioti2022suborn}, these issues arise from the design rationale of timelock-based channel protocols: Lighting~\cite{poon2016bitcoin} and Blitz use a "revoke until time $T$, else execute" logic, while DMC~\cite{decker2015fast} employs a "revoke within time $T$, else execute" logic. In these cases, censoring an honest party's transaction prevents them from revoking misbehavior. Conversely, Thora~\cite{aumayr2022thora} uses an "execute until time $T$, else revoke" logic, and Perun~\cite{dziembowski2019perun} follows a "reply within time $T$, else get punished" logic. In these designs, censoring an honest party's transaction blocks them from executing the committed transactions, resulting in a loss of funds. These censorship attacks, while difficult to trace as excluded transactions are not reported on-chain, can be practically executed using TxWithhold Smart Contracts, as highlighted by BitMEX Research~\cite{bitmex}. Additionally, the cost of such attacks can be significantly reduced since multiple contracts can be targeted simultaneously (e.g., as many as fit in a block)~\cite{avarikioti2024bribe, tsabary2021mad}.

A common approach to lowering the success probability of censorship attacks is to use longer timelocks~\cite{nadahalli2021timelocked}. However, this strategy worsens another significant drawback of timelocks, known as the \emph{griefing attack}~\cite{mazumdar2023strategic}. In such attacks, the receiver initiates a payment path within a PCN but then aborts the payment, causing intermediaries to incur an \emph{opportunity cost} as they lock their assets for the timelock duration without receiving any routing fee. Given the inherent instability of real-world networks and the potential for network outages, relying solely on timelocks for secure and efficient PCN implementation proves to be a flawed approach. This raises a critical question: \emph{Can secure PCNs be designed without the reliance on timelocks?}

\subsection{Related work}

Various solutions have been proposed to counter the rising threat of censorship attacks on Payment Channel Networks (PCNs), many of which exploit miners' incentives, as miners often control censorship. Nadahalli et al.\cite{nadahalli2021timelocked} were among the first to identify the safety vulnerabilities of Hash Timelock Contracts (HTLCs) under censorship, analyzing parameters like timelock duration and transaction fees to determine when HTLCs remain secure. Building on this, Tsabary et al.\cite{tsabary2021mad} introduced MAD-HTLC, a modified structure that incentivizes miners to act as natural enforcers by penalizing malicious actors. However, the counter-bribing vulnerabilities in MAD-HTLC were later identified, leading to the development of He-HTLC~\cite{garewal2020helium} and Rapidash~\cite{chung2022rapidash}. Although these approaches provide secure channel primitives, they rely on specific assumptions, such as mining power distribution, and crucially require parties to remain online and responsive.

Another approach involves introducing third-party entities, known as watchtowers, to manage disputes on behalf of payment channel participants in a timely manner~\cite{mccorry2019pisa,avarikioti2020cerberus,lind2019teechain,avarikioti2018towards}. However, watchtowers are themselves susceptible to censorship attacks~\cite{avarikioti2019brick} unless they also operate as miners, ensuring independent inclusion of transactions. Thus, the security of this approach depends on additional assumptions, such as the watchtower’s mining power and ability to fulfill dual roles. Furthermore, while watchtowers help secure the payment channel primitive, they do not address the execution of multi-hop payments, leaving a gap in overall network protection.

Brick~\cite{avarikioti2019brick} introduced a pioneering payment channel primitive that eliminated the need for timelocks by internalizing dispute resolution within the channel itself. This was accomplished by establishing a committee of watchtowers, known as wardens, who were incentivized to act honestly through the collateral they had locked in the payment channel contract. These wardens were authorized to unilaterally close the Brick channel at the latest update state, removing the need for timelocks. However, extending Brick to support multi-hop payments introduced challenges, as current multi-hop execution mechanisms (e.g., HTLCs) still require timelocks, irrespective of the underlying channel primitive.

Recently, Ersoy et al.\ extended Brick to multi-hop scenarios by proposing a multi-hop payment protocol based on warden committees~\cite{ersoy2022syncpcn}. While promising, this protocol still has notable limitations. First, like traditional multi-hop payment protocols, Ersoy's approach requires the active participation of intermediary parties in each transaction, leading to higher costs and increased latency. Although the protocol claims to support virtual channels, a detailed design is not provided. Second, the protocol lacks a thorough game-theoretic analysis to formally evaluate the security implications of its incentive mechanisms.

Overall, while recent research has made significant progress in addressing timelock-based attacks in PCNs, current solutions remain constrained by specific assumptions (e.g., mining power distribution, intermediary participation) or inherent design limitations (e.g., limited applicability to multi-hop payments). These constraints underscore the need for an alternative approach that overcomes these challenges.

\subsection{Our contribution}
In this work, we introduce \vc, \emph{the first timelock-free virtual channel protocol}, providing an answer to the previously posed challenge. The core concept of \vc revolves around a committee of wardens appointed by the channel parties, who store the channel states and publish the most recent state on-chain when parties are unresponsive, whether due to being offline or censored. These wardens are incentivized with rewards or penalties based on their actions, ensuring the security of \vc. As a result, \vc effectively neutralizes all blockchain liveness attacks, including censorship attacks, within its payment channel network. To enable multi-hop payments, \vc leverages the design principles of virtual channels. First introduced by Dziembowski et al.~\cite{dziembowski2019perun}, virtual channels allow direct off-chain payments without requiring intermediary involvement in every transaction, as seen in multi-hop payment protocols like Blitz. Central to this concept is the mechanism that allows Alice and Bob, who do not share a direct payment channel, to collaborate with Ingrid, who does, to create a virtual channel. This setup enables Alice to transact directly with Bob, bypassing the need for Ingrid's participation. In essence, a virtual channel is a channel built atop two existing payment channels. \vc adopts a similar architectural design, constructed over two Brick channels, rather than relying on timelock-based payment channel primitives.

However, constructing a secure virtual channel atop two Brick channels presents several challenges due to Brick's design and inherent limitations. Specifically, a Brick channel relies on complex incentive mechanisms for both parties and the warden committee to maintain balance security (i.e., safety) and prevent hostage situations (i.e., liveness). These incentive mechanisms vary between Brick channels with different warden committees, making them not directly composable. As a result, the virtual channel cannot function as a standard payment channel without additional considerations.

Specifically, three major challenges arise when transitioning from Brick to \vc. First, Brick employs proofs-of-fraud that allow a party to claim the wardens' collateral. However, since the virtual channel comprises two separate Brick channels, a party in one channel cannot directly access the collateral locked in the other channel's smart contract. Furthermore, in a virtual channel, the intermediary party remains unaware of the virtual channel's state, as transactions are exclusively exchanged between the other two parties. This leaves the intermediary party vulnerable to potential fraud by the other participants or wardens during the channel closure. Therefore, it is uncertain whether the proof-of-fraud mechanism can still ensure the security of all virtual channel participants, including the intermediary. Second, since the virtual channel is built directly on payment channels without requiring additional on-chain deposits for its opening, ensuring compatibility with the underlying Brick payment channels presents a challenge. Third, although Brick is claimed to be secure in the rational setting, no formal framework or analysis has been provided. Thus, the third challenge is to develop a formal model that can prove \vc's security when all participants act rationally.


To address the first challenge, we modify the protocols involved in the update and close operations of \vc. We ensure the correctness of channel closure while allowing parties to penalize only those wardens with collateral in the relevant payment channel's smart contract. Additionally, we achieve closing security for the intermediary party by ensuring they either have \vc knowledge before closing the channel or the closing party is indistinguishable in terms of possessing such knowledge. For the second challenge, we draw inspiration from~\cite{aumayr2021bitcoin}, offloading the virtual channel to an on-chain payment channel only in pessimistic cases with the blockchain’s assistance. This allows us to distribute wardens' collateral according to each channel's balance, ensuring \vc remains securely compatible with the underlying payment channel. To address the third challenge, we model our protocol using Extensive Form Games (EFG), drawing on recent work that introduced the first game-theoretic analysis suitable for off-chain protocols like the Bitcoin Lightning Network~\cite{rain2022towards}. Specifically, we model \vc’s closing operation as an EFG and prove the game-theoretic security of the protocol by analyzing all possible strategies during the channel closure.

\noindent\textbf{Summary of Contribution.}
Our contributions can be summarized as follows:
\begin{itemize}
    \item We introduce \vc, the first timelock-free virtual channel protocol~\cite{dziembowski2019perun}. \vc builds upon the asynchronous payment channel primitive, Brick~\cite{avarikioti2019brick}, enabling secure payments between parties that do not share a direct channel, all without the need for timelocks (Section~\ref{sec:design}). We begin by presenting the design and security analysis of \vc with a single intermediary and then extend the protocol to support multi-hop payments involving more than two hops (Section~\ref{sec:multi}).
    
    \item We formalize the security properties of \vc and conduct a security analysis within the honest/Byzantine model (Section~\ref{sec:BAanalysis}). Specifically, we demonstrate that our protocol is secure under the assumption of distrusting participants, where $f$ out of $3f+1$ wardens in each payment channel are Byzantine, and the remainder are honest.
   
    \item We design the incentive mechanisms of \vc and introduce a formal game-theoretic model to represent the protocol (Section~\ref{sec:GTanalysis}). We prove that following the \vc closing protocol honestly constitutes a Subtree Perfect Nash Equilibrium (SPNE) strategy, ensuring the protocol’s game-theoretic security in the rational model. To our knowledge, this is the first off-chain protocol to include a formal game-theoretic analysis. 
   
    \item We evaluate the practicality of \vc by fully implementing its on-chain functions on the Ethereum blockchain using Solidity (Section~\ref{sec:eva}). We assess the gas costs for each procedure with 10 wardens in each committee. While the cost of opening \vc is zero, deploying and opening the underlying payment channel requires 444,4861 gas (approximately 14.83 USD). In the pessimistic scenario, closing \vc along with the underlying payment channel costs 2,079,766 gas (around 6.94 USD). We compare these costs with the timelock-based virtual channel protocol Perun: opening \vc's underlying channel costs 1.5 times more than Perun, while the pessimistic closing costs are 3 times higher. This shows that the additional cost of eliminating timelocks is not excessive. Additionally, we evaluate how gas fees scale as the number of wardens per committee increases from 10 to 25.
\end{itemize}

\section{Background}
\label{sec:background}
In this section, we provide a brief overview of the fundamental concepts of payment channels, with further details available in Gudgeon et al.\cite{gudgeon2020sok}. We then explore the design of the asynchronous Brick payment channel\cite{avarikioti2019brick}, which serves as the foundation of our protocol. Finally, we discuss payment channel networks and virtual channels.

\subsection{Payment channels}

A payment channel allows two users to exchange arbitrary transactions off-chain, while only a constant number of transactions are recorded on-chain, such as a worst-case scenario of three transactions in Lightning~\cite{poon2016bitcoin}. To initiate a payment channel, parties deposit coins on the blockchain, which can only be spent when specific conditions are met, such as requiring both parties' signatures (\textit{channel open}). Once the coins are locked on-chain, the parties can communicate off-chain and update the channel's balance or state (\textit{channel update}). To claim their funds on-chain, the parties can close the payment channel by publishing the most recent agreed-upon balance (\textit{channel close}). In the following, we discuss the key steps involved in these three operations in detail.

\textbf{Channel open.} Consider two parties, Alice and Bob, who wish to open a payment channel with initial deposits of $x_A$ and $x_B$ coins, respectively. For cryptocurrencies that support smart contracts, Alice and Bob can publish a \textit{payment channel smart contract} $C_L$ on the blockchain, containing a total balance of $x_A + x_B$ from both parties. These coins can only be spent with the signatures of both Alice ($\sigma_A$) and Bob ($\sigma_B$). Once $C_L$ is published on-chain, the payment channel is officially open.

\textbf{Channel update.} Suppose Alice wants to pay Bob an amount $a \leq x_A$. Alice generates a new payment channel state, $s_L$, signs it with her private key, and sends $\{s_L,\sigma_A(s_L)\}$ to Bob as an update request. Bob then verifies the request. If he agrees with the new state, he signs it as well and returns the final state, $\{s_L,\sigma_A(s_L),\sigma_B(s_L)\}$, to Alice. If Bob disagrees, he simply ignores the request. Once a valid channel state, signed by both parties, is generated (along with any other predefined protocol data exchange, such as revocation keys in Lightning), the new off-chain transaction between Alice and Bob is considered successfully completed.

\textbf{Channel close.} A payment channel can be closed either collaboratively or unilaterally by one party. In a collaborative closure, both parties publish a state with their signatures on-chain and distribute the channel balance accordingly. In a unilateral closure, one party publishes the most recent payment channel state signed by both, $\{s_{latest}, \sigma_A(s_{latest}), \sigma_B(s_{latest})\}$. After verifying the signatures, the smart contract $C_L$ closes the channel and distributes the coins according to the submitted state. The key challenge arises when a malicious party attempts to close the channel using an outdated state. Synchronous payment channel protocols typically address this by enforcing a timelock on the submitted state, which corresponds to the party posting it on-chain. This timelock, often combined with a secret exchanged during the channel update, enables the counterparty to punish a malicious party attempting to use an old state. Alternative techniques, such as verifiable timed signatures~\cite{sri2022vts}, can also be used, but they still rely on timing assumptions.

\subsection{Brick channel}
\label{sec:Brick}

The Brick channel, introduced in \cite{avarikioti2019brick}, is an asynchronous payment channel primitive. To address the challenges of operating without timelocks, Brick incorporates a committee of third-party entities known as wardens in the channel operations. In essence, wardens are responsible for verifying and storing commitments of channel state updates, which can then be used by a channel party to unilaterally close the Brick channel. Importantly, wardens are not fully trusted; rather, they are incentivized to follow the protocol honestly. Compared to its synchronous counterparts, Brick exhibits the following key differences:

\textbf{Channel structure.} In addition to the two primary parties, Alice and Bob, a Brick channel includes a committee of $3f+1$ wardens, where up to $f$ wardens can behave maliciously (Byzantine). The basic structure of the Brick channel is illustrated in Fig.~\ref{fig:brick}.

\begin{figure}[htbp]
    \centering
    \includegraphics[width=0.7\columnwidth]{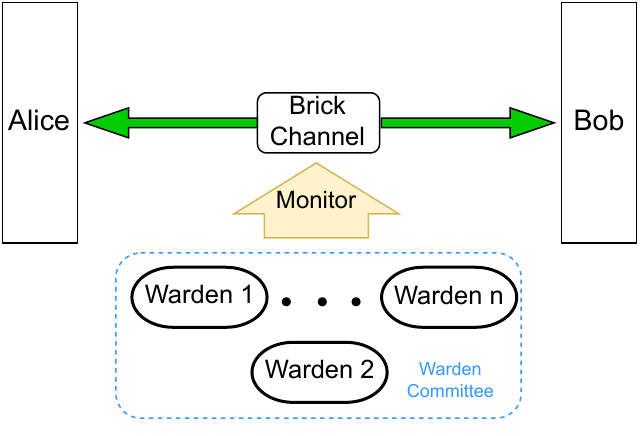}
    \caption{Brick payment channel}
    \label{fig:brick}
    \end{figure}

\textbf{Brick open.} Once Alice and Bob agree to open a Brick channel, they broadcast the channel information to the wardens. Each party can consider the channel open only after receiving at least $2f+1$ signed acknowledgments from the wardens -- a quorum certificate. The threshold of $2f+1$ ensures safety in asynchronous communication networks.

\textbf{Brick update.} Each state generated by Alice and Bob is hashed and signed by both parties. The state is then assigned a sequence number to indicate its order. This sequence number, along with the hashed state, is signed and broadcast to the wardens. Once at least $2f+1$ wardens have signed the sequence number, the state is considered valid, and the parties can execute it.

\textbf{Brick close.} The Brick channel closing protocol has two scenarios: \textit{optimistic} and \textit{pessimistic}. In the optimistic case, both parties are online and responsive, allowing them to collaboratively generate and sign a closing request, which is then published on-chain. In the pessimistic case, one party, say Alice, may be offline for an extended period and unresponsive to Bob's close request. To prevent the channel from being indefinitely locked, the Brick protocol allows Bob to close the channel unilaterally. However, Bob alone does not have a valid closing request, as during the channel update, the parties only exchange signatures on the hashed state, which is insufficient to close the channel. To close unilaterally, Bob initiates a closing request on-chain, prompting the wardens to publish the latest stored sequence number. The valid closing state is defined by the highest sequence number signed by Alice, Bob, and at least one warden. To incentivize honest behavior from the wardens, Brick employs a punishment mechanism based on proofs-of-fraud. A \emph{proof-of-fraud} consists of a warden’s signature on an update with a higher sequence number than the one they submitted for the closing request. Bob collects these signatures during the channel update process. If a valid \emph{proof-of-fraud} is provided, Bob can claim the corresponding warden's collateral on-chain. According to Brick's security analysis, if the aggregate collateral of wardens submitting outdated states exceeds the channel balance, Bob has no incentive to close the channel incorrectly. Instead, he would claim the wardens' collateral and award the entire channel balance to his counterparty. Specifically, for a Brick channel with $v$ coins as the balance, each of the $3f+1$ wardens must deposit at least $\frac{v}{f}$ coins on-chain to ensure security.

\subsection{Virtual channels}

 \textit{Virtual channels}\cite{dziembowski2019perun,aumayr2021bitcoin} offer a solution for executing transactions between parties who do not share a direct channel while minimizing the involvement of an intermediary (Ingrid). Specifically, Alice and Bob can establish a virtual channel on top of two existing payment channels, as shown in Fig~\ref{fig:virtualchannel}. Similar to payment channel protocols, virtual channels typically involve three operations: \textit{open}, \textit{update}, and \textit{close}. However, Ingrid only needs to participate in the \textit{open} and \textit{close} operations. For example, if Alice and Bob want to create a virtual channel with a balance of $x_A + x_B$ coins, not only do Alice and Bob contribute coins, but Ingrid also deposits $x_B$ and $x_A$ coins, respectively, to open the virtual channel, as shown in Fig~\ref{fig:virtualchannel}. Importantly, no additional coins are deposited on the blockchain for the virtual channel; instead, the virtual channel's funding is achieved by updating the payment channel state and utilizing a portion of the payment channel's funds. Unlike payment channels, opening a virtual channel generally occurs off-chain. After the virtual channel is opened, Alice and Bob can transact directly with each other during the \textit{update} operation without needing Ingrid’s involvement. When Alice and Bob decide to close the virtual channel, Ingrid participates in updating or closing the two underlying payment channels. The close operation only affects the underlying payment channels without directly impacting the blockchain. Accessing the blockchain is only required in pessimistic situations, such as when one party becomes unresponsive. Therefore, virtual channels can be seen as a "Layer 3" solution in blockchain architecture.

 \begin{figure}[htbp]
    \centering
    \includegraphics[width=\columnwidth]{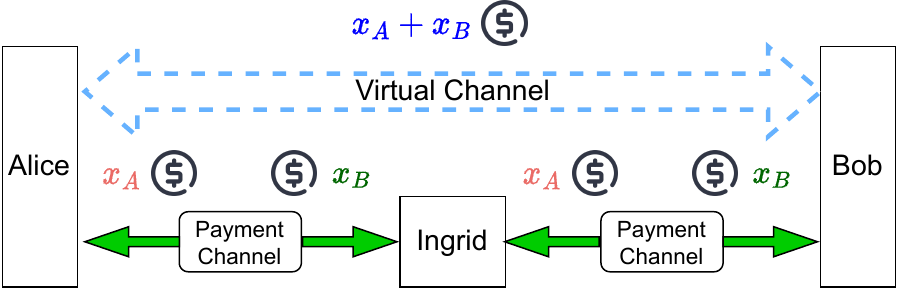}
    \caption{Virtual channel structure}
    \label{fig:virtualchannel}
    \end{figure}

\section{Model}
\label{sec:model}

In this section, we first present our system model and underlying assumptions. We then formally define two distinct threat models: one assumes Byzantine participants, while the other assumes participants are rational agents aiming to maximize their profit. Finally, we outline the desired security properties of \vc under both models.

\subsection{System model and assumptions} 
There are two types of participants in \vc: (1) the \textit{main parties}, which include Alice ($A$), Bob ($B$), and Ingrid ($I$) in the two-hop scenario; and (2) the \textit{wardens}, organized into warden committees $C_A$ and $C_B$. The three main parties maintain two underlying timelock-free payment channels (TPC), such as Brick. The Alice-Ingrid and Ingrid-Bob channels are each monitored by the wardens in $C_A$ and $C_B$, respectively. The general structure of \vc is illustrated in Fig.~\ref{fig:virtualbrick}.

We assume that all participants are computationally bounded and that cryptographically secure communication channels, signatures, hash functions, and encryption mechanisms are in place. The communication channels between \vc participants are considered asynchronous, meaning messages from honest parties may be delayed indefinitely but will eventually be delivered. This includes messages sent to the blockchain (miners), which can also experience arbitrary delays. We further assume that the underlying blockchain is safe and guarantees persistence~\cite{garay2024bitcoin}.

\subsection{Threat models}
We introduce two threat models: the conventional security model, which assumes the presence of Byzantine participants, and a game-theoretic model, which considers rational participants. A game-theoretic analysis is particularly important for off-chain protocols, as the standard Byzantine security model may fail to capture certain threats that arise from party collusion, such as the WormHole attack~\cite{malavolta2018anonymous}. Given the significant impact of Byzantine behavior on participant utilities and the lack-to the best of our knowledge—of a unified framework that addresses both models, it is crucial to assess the security of our protocol under both perspectives.

\textbf{Byzantine security model.} In this model, we assume each committee consists of $3f+1$ wardens, of which $f$ may be Byzantine, while the rest are honest and follow the protocol as specified. We note that the protocol's security can be extended to scenarios where the number of wardens differs as long as the Byzantine ratio remains no larger than $\frac{f}{3f+1}$ (i.e., the total number of wardens may vary depending on the setup, but the Byzantine proportion must be maintained). In addition to wardens, the security model also accounts for cases where the main \vc parties may be Byzantine, ensuring security for any honest participant. This is the weakest meaningful assumption, as collusion between all parties against a single honest party could occur; however, considering all parties deviating from the protocol holds no value since security properties apply only to honest participants. Byzantine nodes can arbitrarily deviate by delaying, crashing, or signing/publishing incorrect states or messages. However, they cannot drop honest messages or forge signatures, as these actions violate network and cryptographic assumptions.

\textbf{Game-theoretic model.} 
In this model, all participants are treated as rational, mutually distrusting agents. Rational agents will deviate from the protocol—such as double signing, publishing outdated states, or crashing—when doing so allows them to maximize their utility, i.e., gain more profit. Wardens are required to lock collateral for their participation in each channel. Overlapping wardens lock collateral in separate smart contracts for each channel, ensuring that punishments remain independent. Additionally, the participants' budgets are limited to their on-chain collateral, and external budgets are not considered in our analysis. This approach is consistent with most blockchain protocols, as proving economic security in the presence of infinite shorting markets (e.g., Bitcoin or Ethereum) is infeasible.

\subsection{Protocol Goals}

The two primary goals of our protocol are \emph{balance security} and \emph{liveness}. Balance security ensures the basic safety of channels, meaning that no honest participant loses coins during the execution of the virtual channel protocol. Honest participants are defined as those who adhere to the protocol specifications.

\begin{definition}[Balance security]
\label{def:security}
Any honest participant of \vc does not lose coins.    
\end{definition}

In addition to balance security, \vc must ensure that the protocol progresses meaningfully. For example, consider a channel protocol that only contains dummy functions that make no changes to the channel state. Although no participant would lose coins, such a protocol is useless, as no valid channel state could ultimately be committed on-chain. More broadly, hostage situations can still cause losses for participants, even if the channel funds remain intact. The liveness property addresses these issues, complementing balance security by ensuring meaningful progress.

\begin{definition}[Liveness]
\label{def:liveness}
Any valid operation (update or close) on the state of the virtual channel, involving at least one honest participant (Alice, Bob, or Ingrid), will eventually either be committed on-chain or invalidated.
\end{definition}

We note that the \emph{validity} of operations is determined by the protocol specification. For example, in our protocol, an operation is considered valid if it is agreed upon by the two main parties, involving at least one honest participant (Alice, Bob, or Ingrid). In contrast, in the Lightning Network, a valid update corresponds to the so-called commitment transaction~\cite{poon2016bitcoin}. A valid update is either the latest one, capable of being committed on-chain, or it is replaced by a newer valid update and thereby invalidated.

In our analysis, we demonstrate that \vc satisfies these properties in both the Byzantine and game-theoretic models. However, to prove these properties in the game-theoretic model, we employ different tools: we model the protocol as an Extensive Form Game (EFG) with \emph{perfect information} and show that the correct strategy profile forms a Subgame Perfect Nash Equilibrium (SPNE). Detailed definitions and the security properties encapsulating these concepts are provided in Appendix~\ref{sec:gamemodel}.

\section{Protocol Overview}
\label{sec:overview}

The core structure of \vc is illustrated in Fig.~\ref{fig:virtualbrick}. Despite its apparent simplicity, several challenges arise in securely designing the protocol: a) unlike the two-party TPC, \vc involves three main parties, introducing the possibility of collusion between two parties (and the wardens) to defraud the third party; b) the wardens in \vc can only be penalized in the TPC if they have deposited collateral, even though they may verify transactions that affect both underlying channels; and c) \vc transactions are only known to Alice and Bob, leaving Ingrid unaware of the most recent state of the virtual channel. In the following, we present \vc using a strawman approach.

 \begin{figure}[htbp]
    \centering
    \includegraphics[width=\columnwidth]{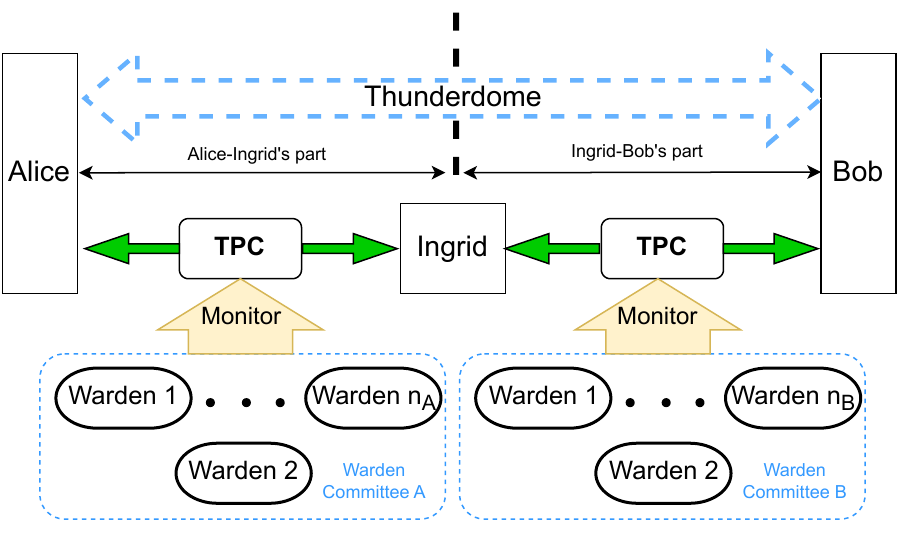}
    \caption{\vc structure}
    \label{fig:virtualbrick}
    \end{figure}

\textbf{Naive design.}
In a naive design of \vc, we follow the same rationale as regular TPC protocols like Brick. When opening the virtual channel, the wardens of each TPC lock a portion of their collateral into \vc. For example, if 2 out of 10 coins in a TPC are locked in the virtual channel, each warden locks 20\% of their collateral during the funding transaction. The update phase of \vc operates similarly to Brick. However, the closing phase differs during collaborative closing: \vc cannot be closed solely by Alice and Bob; Ingrid’s agreement is also required. For unilateral closing, the smart contract only requires a quorum certificate ($2f+1$ signatures) from the closing party’s warden committee (e.g., $C_A$ or $C_B$).

However, the security of this protocol can be easily compromised if two parties collude. Suppose Alice and Bob collude to cheat Ingrid. Since the unilateral closing only requires the quorum certificate from each warden committee independently, Alice and Bob can create and broadcast two different update states—one to each warden committee, $C_A$ and $C_B$. For instance, in one state, Alice holds all the coins, while in the other, Bob holds all the coins of the virtual channel. During the \vc unilateral closing procedure, neither Alice nor Bob will publish any proof-of-fraud, even if the wardens from different committees submit different states on-chain, as each warden only signs one state and publishes it honestly. Consequently, \vc will close with two different states in the underlying TPCs, causing Ingrid to lose coins. A straightforward solution would be to require the wardens to send their signatures to Ingrid as well. However, under asynchronous conditions, it becomes impossible to differentiate between a Byzantine warden who deliberately withholds the signature from Ingrid (but sends it to Alice) and an honest warden whose signature simply doesn't arrive in time.

\textbf{Smart contract cross-checking.}
To prevent attacks and protect Ingrid in \vc, we introduce a \emph{smart contract cross-checking} scheme during the unilateral closing procedure. Specifically, after a smart contract receives enough states from wardens that have not been proven fraudulent, it notifies the other contract of the latest state it has learned and queries the other contract’s latest state by sending a transaction. Both contracts then decide on the state with the higher sequence number, which will be used as the final closing state.

However, this approach remains vulnerable to certain safety attacks. For example, if Alice and Bob collude to generate two different states with the same sequence number, they can broadcast these states to their respective warden committees and initiate unilateral closing simultaneously. In the worst-case scenario, both smart contracts may send transactions to notify each other of the states they have received, and these transactions could be included in the same block, arriving simultaneously. If a naive scheme such as "only storing the state from the counterparty contract" is applied in this situation, the two parts of \vc could still be closed with different states, resulting in potential losses for Ingrid.

\textbf{Use leader contract to solve collision.}
To address the potential collision when transactions from two contracts arrive simultaneously, we designate one contract as the \emph{leader contract}. In the event that two different states have the same sequence number, both contracts will store only the state from the \emph{leader contract}. The leader contract must be agreed upon by all parties before the virtual channel is opened and should notify the corresponding smart contract as required.

Once we establish how smart contracts should handle states published by wardens, the next question is how many on-chain publications are necessary, given that asynchrony only guarantees eventual message delivery. If an updated state is confirmed with only $2f+1$ signatures from each warden committee, there could still be up to $f$ (honest) wardens who have not received the latest state, even after the protocol was closed (incorrectly). To mitigate this, we require the smart contract to wait for at least $f+1$ publishments from the corresponding committee before closing the channel. This ensures that at least one warden with locked collateral knows the latest state.

However, if we consider rational, profit-maximizing players, Bob could simply bribe that specific warden and close the channel using an outdated state. Since the collateral held by each warden in TPC is much smaller than the total channel value (which could be equal to the value locked in \vc), Bob could benefit by bribing the warden and closing the channel in a more favorable state. Thus, this protocol would not be secure in the rational model.

\textbf{Ensuring enough collateral.}
A straightforward solution to this bribing problem is to require each warden to lock collateral equal to the full channel value. However, this approach is impractical due to over-collateralization. Additionally, either the funding transactions of the underlying TPCs would need to be topped up to support future virtual channels, or additional on-chain transactions would be required when a virtual channel opens. These implications are undesirable, so we aim for a solution that does not require additional collateral.

The first step is determining the necessary amount of collateral. Let the balance of \vc be $v$, and assume the parties have locked $v_i$, $i \in \{A, I_A, I_B, B\}$ coins in the Alice-Ingrid and Ingrid-Bob payment channels, respectively. Since Ingrid participates in both channels, she locks coins in both, and the locked balances must satisfy the relationship: $v_A + v_B = v_{I_A} + v_{I_B} = v_A + v_{I_A} = v_{I_B} + v_B = v$. Any collusion set could potentially gain up to $v$ coins. On the one hand, a single Byzantine party could cheat their counterparty and gain at most the number of coins locked by the counterparty, which is no more than the \vc balance $v$. On the other hand, collusion between Alice and Bob could cheat Ingrid, who has locked $v_{I_A} + v_{I_B} = v$ coins, or Alice (or Bob) could collude with Ingrid to cheat Bob (or Alice), who has locked $v_A$ (or $v_B$) coins, which are also less than or equal to $v$. Therefore, to guarantee security, we must ensure that the total collateral of wardens that can be punished in the event of fraud is no less than \vc's balance $v$.

Thus, we require at least $2f+1$ wardens from each committee, $C_A$ and $C_B$, to publish their closing state in \vc's unilateral closing protocol. Combined with the quorum certificates required to update the \vc state, we ensure that at most $f$ wardens per committee can be slow. By waiting for $2f+1$ publications from accountable wardens, at least $f+1$ publications can be punished if they submit outdated information. Given that each warden locks $v/f$ coins in \vc, the total collateral that can be claimed per channel in case of fraud is $(f+1) \times \frac{v}{f} > v$.

\section{\vc Design}
\label{sec:design}
In this section, we present the detailed \vc protocol, which consists of three operations: \textit{Open}, \textit{Update}, and \textit{Close}. We assume the existence of two consecutive underlying TPCs: Alice-Ingrid's and Ingrid-Bob's. Each channel is managed by an on-chain smart contract, which handles the locked coins of both channel parties as well as the collateral of the wardens in the corresponding committee. The smart contract only has information about the specific payment channel it governs and has no knowledge of other payment channels or virtual channels.

\subsection{\vc Open}
\label{sec:open}
To initiate the opening of \vc, agreement among all three main parties is required. As evidence for the channel opening, our protocol mandates that all three parties collaboratively generate a unique transaction, referred to as the "register transaction," $TX_r$. This transaction includes the signatures of all three parties, serving as proof of their unanimous consent to open the virtual channel. Additionally, $TX_r$ contains essential information related to the virtual channel, such as its initial state, total channel balance, and \emph{contract information}, including the contract address and leader contract identifier. Under optimistic conditions, $TX_r$ remains off-chain and is not published on the blockchain. However, in pessimistic scenarios, such as when some parties go offline, the remaining online parties may involve the payment channel smart contract to help close \vc. In this case, $TX_r$ will be submitted to the smart contract and used to facilitate decisions during the \vc unilateral closing process when some parties are offline. Protocol~\ref{vopen} outlines the process of opening a \vc channel.

\begin{algorithm}[h]
    \caption{\vc Open}
    \begin{algorithmic}[1]
    \label{vopen}
        \REQUIRE Parties $A$, $I$ and $B$, TPC committees $C_A$ and $C_B$, initial virtual channel state $s_1$, total channel balance $v$ and contract information $ci$. 
        \ENSURE Virtual channel register transaction $TX_r$ and open a virtual channel
        
        \textcolor{magenta}{/*Combine wardens to form a virtual channel committee*/}
        \STATE $A$, $I$ and $B$ combine wardens from $C_A$ and $C_B$ to form virtual channel committee $C_V=C_A \cup C_B$.
        
        \textcolor{magenta}{/*Parties reach agreement on the open state*/}
        \STATE $A$ and $B$ generate two pre-register transaction with their signatures $TX_r'=\{pk_A, pk_B, pk_I,\{pk_{C_V}\}, s_1, v, ci, \sigma_A/\sigma_B\}$ and send to $I$.
        \STATE After receiving two transactions, $I$ responds with $TX_r''=\{pk_A, pk_B, pk_I,\{pk_{C_V}\}, s_1, v, ci,\sigma_A/\sigma_B, \sigma_I\}$.
        
        \STATE $A$ and $B$ exchange responses of $I$ and generate final register transaction with all partis' signatures $TX_r=\{pk_A, pk_B, pk_I,\{pk_{C_V}\}, s_1, v, ci, \sigma_A ,\sigma_B, \sigma_I\}$.
        \STATE All parties broadcast $TX_r$ to wardens in virtual channel committee $C_V$, and receive more than $t$ responses with wardens' signatures, then the virtual channel is considered to be opened.
        \STATE $A$ and $I$ ($I$ and $B$) update payment channels according to initial virtual channel state $s_1$.
    
    \end{algorithmic}
\end{algorithm}

First, all three parties—Alice, Ingrid, and Bob—must select a new committee for \vc based on the two existing payment channel committees, $C_A$ and $C_B$. A straightforward approach is to combine the two committees into one. While selecting a subset of wardens from each committee could reduce message complexity, it introduces additional security concerns, such as ensuring the safety of the selected virtual channel committee, which is beyond the scope of this paper. Therefore, to simplify the analysis, we assume the virtual channel committee $C_V$ is the combined set of the two payment channel committees, $C_A$ and $C_B$.

Once $C_V$ is finalized, the three main parties collaboratively generate the register transaction $TX_r$. To ensure that all parties eventually have the same $TX_r$, or none at all, Alice and Bob first generate and send two pre-register transactions with their signatures to Ingrid. Upon receipt, Ingrid verifies the transactions match, and sends the new pre-register transaction back to Alice and Bob. Alice and Bob then exchange response messages to collect each other’s signatures. Finally, all three parties obtain the complete $TX_r$, which includes the public keys of the wardens in $C_V$, the public keys of Alice, Bob, and Ingrid ($pk_A$, $pk_B$, and $pk_I$), the initial state $s_1$, and the signatures of the three parties ($\sigma_A$, $\sigma_B$, and $\sigma_I$). After generating the register transaction $TX_r$, the main parties broadcast it to all wardens in $C_V$ for verification.

The final step in opening \vc, after all parties reach an agreement, is to "virtually lock" the coins for the virtual channel. Unlike payment channels, which require publishing funding transactions on-chain, the parties of both underlying payment channels—Alice and Ingrid, and Ingrid and Bob—update their respective payment channels to lock coins for use in the virtual channel. For example, if the payment channel state between Alice and Ingrid is $(a, b)$, and Alice wants to open a virtual channel with Bob with an initial state of $(c, d)$, then Alice and Ingrid must update their payment channel state to $(a - c, b - d)$. This means Alice deposits $c$ coins and Ingrid deposits $d$ coins for the virtual channel. This update locks the coins ($c + d$) for exclusive use in the virtual channel, while the wardens' collateral is allocated based on the balance of both the payment channel and the virtual channel. The virtual channel is considered open once the payment channel has been updated accordingly. It's important to note that opening and closing are the only operations that require changes to the underlying payment channels. Updates to \vc and the payment channels can occur in parallel, as long as the coins locked for the virtual channel are not used during the payment channel update.

\subsection{\vc Update}
\label{sec:update}

The \textit{update} procedure of \vc is outlined in Protocol~\ref{alg:vupdate}. After Alice and Bob agree on the new \vc state $s_i$ and sequence number $i$, they generate an announcement containing their signatures, $M = \{s_i, i, \sigma_A(s_i, i), \sigma_B(s_i, i)\}$, and broadcast it to all wardens in the \vc committee $C_V$. To limit the influence of slow wardens in each committee, both parties must wait for the quorum certificate from each warden committee ($C_A$ and $C_B$) before proceeding with the next state update. Although Byzantine or rational wardens could deviate from this protocol by broadcasting different states to wardens or by not waiting for sufficient responses, our closing protocol, which will be introduced in the next section, ensures that these deviations will not compromise the protocol's security.

\begin{algorithm}[!h]
    \caption{\vc Update}
    \label{alg:vupdate}

    \textbf{Main Parties}
    \begin{algorithmic}[1]
        \REQUIRE Parties $A$ and $B$, warden committee $C_V$ comprising $C_A$ and $C_B$, current state $E_{P_I}(s_i)$.
        \ENSURE Update \vc to a new valid state.

        \STATE Both parties $A$ and $B$ sign and exchange the new virtual channel state: $M=\{s_i, i , \sigma_A(s_i,i), \sigma_B(s_i,i)\}$.
        
        \STATE $A$ and $B$  broadcast $M$ to committee $C_V$. \textcolor{magenta}{/*Along with a fee $r$*/}

        \STATE $A$ ($B$) waits for a quorum certificate ($2f+1$ signatures) from each warden committee, then $A$ ($B$) execute the virtual channel state $s_i$.

    \end{algorithmic}

    \textbf{Wardens}
    \begin{algorithmic}[1]
    
    \REQUIRE Parties $A$ and $B$, warden committee $C_V$ comprising $C_A$ and $C_B$, sequence number $i$.
    \ENSURE $C_V$ updates to a new valid state.
        
         \STATE Each warden $W_j$, upon receiving $M$, verifies that both parties’ signatures are present and the sequence number is exactly one higher than the previously stored sequence number. If the warden has published a state or has signed an $M$ or $M'$ with the same sequence number, it ignores the state update. Otherwise, $W_j$ stores $M$ (as a possible update, not yet replacing the $i-1$-th committed state), and sends its signature $\sigma_{W_j}(M)$ to both parties. \textcolor{magenta}{/*Only to parties that paid the fee*/}
     \end{algorithmic}

\end{algorithm}
 
\subsection{\vc Close}
\label{sec:close}

The final operation of \vc is the \textit{close}, which is initiated when at least one participant wishes to close the virtual channel. This can be optimistically executed by unlocking the previously virtually locked coins in the underlying TPCs through a simple channel update, provided that all parties are online and responsive. However, to ensure that coins are not locked indefinitely if some participants are unresponsive, \vc includes a mechanism that allows any party to unilaterally initiate the closing of the virtual channel on-chain, using the wardens. In the following, we outline the different subroutines for the \vc close operation and explain the rationale behind the design, starting from the most optimistic scenario.

\subsubsection{Optimistic situation}
In the optimistic scenario, all three parties—Alice, Bob, and Ingrid—are online and acting honestly. In this case, \vc can be closed by simply updating the underlying payment channels. Alice and Bob first send closing requests to Ingrid. Once Ingrid receives matching announcements from both parties, she verifies that the requests align. The parties then collaborate to update their payment channels according to the final virtual channel state.

\begin{algorithm}[!h]
    \caption{\vc Closing protocol}
    \label{alg:close-3online}
    \begin{algorithmic}[1]
        \REQUIRE Closing party $P\in\{A,I,B\}$ and payment channel counter party $Q\in\{A,I,B\}$, closing request sent by party $m\in\{P,Q\}$: $VS_m=\{s_{i_m}, i , \sigma_m(s_{i_m},i), \sigma_m(s_{i_m},i)\}$. 
        
        \ENSURE Close the virtual channel when all parties are online and responsive.
        
         \textcolor{magenta}{/*Parties send closing request*/}
        \STATE Closing party $P$ sends closing requests to others.
        
        \textcolor{magenta}{/*Counterparty is offline or sends incorrect request*/}
        
        \IF{$Q$ is offline or $VS_Q\neq VS_P$}
        \STATE $P$ runs Protocol~\ref{alg:uni-AB} separately to close the channel with $Q$

        \ELSE 
        \STATE $P$ update the payment channel with $Q$
        \ENDIF

        \STATE If both parts of \vc are closed, the virtual channel is successfully closed.
    \end{algorithmic}
\end{algorithm}

However, the optimistic closing procedure is insufficient on its own for several reasons. First, the closing requests from Alice and Bob may differ, leading to inconsistencies. Additionally, any of the three parties could be offline, preventing them from sending or responding to closing requests and resulting in a deadlock. To address these issues, we ensure the correct closing of \vc by allowing any online party to unilaterally close the channel. This protocol handles cases where counterparty agreement cannot be obtained due to an asynchronous network, offline status, or Byzantine behavior. To facilitate unilateral closure in \vc, we leverage the assistance of wardens, as described below. Crucially, the initiating party does not need to determine whether the counterparty is offline or if their message has not arrived; the mere absence of a response triggers the unilateral closing process.

\subsubsection{Pessimistic situation}

\textbf{Ingrid is offline or misbehaves.}
The first situation occurs when Ingrid is offline while Alice and Bob wish to close the virtual channel, or when Ingrid responds inconsistently to Alice's and Bob's honest closing requests. In this case, \vc allows Alice and Bob to close the channel unilaterally, requiring them to interact with the smart contract, effectively offloading the virtual channel closure to the underlying payment channel. Protocol~\ref{alg:uni-AB} outlines the steps from Alice's perspective.

The core idea of Protocol~\ref{alg:uni-AB} is for the warden committee $C_V$ to publish the latest channel state, enabling the closing party to punish malicious wardens by slashing their collateral in the smart contract. At the start of Protocol~\ref{alg:uni-AB}, the closing party submits the register transaction $TX_r$ to the smart contract, allowing the payment channel contract to verify published states and proofs-of-fraud using the public keys of the \vc parties and wardens from the respective committees included in $TX_r$. With $3f+1$ wardens in each committee, the closing party publishes proofs-of-fraud only after at least $2f+1$ wardens from their committee have published closing states, ensuring sufficient collateral for penalties.

As discussed in Section~\ref{sec:update}, it is possible for parties to deviate from the update protocol. As a result, when Alice and Bob execute Protocol~\ref{alg:uni-AB} simultaneously, the closing states published by the wardens could differ. To address this potential security issue, we require smart contracts to cross-check the states published by the wardens, as shown in Protocol~\ref{alg:crosschecking}. In this process, the two contracts exchange the latest state they have received and only retain the one with the higher sequence number. If two states have the same sequence number but different values, only the state from the \emph{leader contract}, which is pre-determined in the register transaction $TX_r$ during the channel opening, will be stored. Ultimately, \vc will be closed using the stored state after cross-checking, if there are not enough fraud proofs to invalidate it.

In \vc, parties can close the virtual channel unilaterally only by simultaneously closing the payment channel, as the latter becomes redundant if Ingrid does not respond to the \vc closing request. Furthermore, \vc enforces the virtual channel to be closed before the payment channel. 
To close both channels, parties first interact with the TPC's smart contract to verify the virtual channel's existence. The virtual channel is then closed through the combined warden committees, followed by the payment channel's closure. 
If a party attempts to close the payment channel without settling the virtual channel first, the \vc balance is forfeited to the payment channel counterparty as a penalty.

\begin{algorithm}[!h]
    \caption{Unilateral closing}
    \label{alg:uni-AB}
    \begin{algorithmic}[1]
        \REQUIRE Party $A$, wardens $W_1, \cdots, W_{n_A}$ from payment channel committee $C_A$, virtual channel state $VS$, virtual channel register transaction $TX_r$. 
        \ENSURE Close the payment channel and one part of the virtual channel.
        
        \textcolor{magenta}{/*Alice registers virtual channel*/}
        \STATE $A$ publishes the register transaction $TX_r$ on-chain and prove the validity of virtual channel.

        \textcolor{magenta}{/*Wardens publish virtual channel states */}
        \STATE Each warden $W_j$ publishes on-chain signature lists on the stored announcements for payment channel and virtual channel $\{VS,\sigma_{W_j}(VS)\}$.
        
        \textcolor{magenta}{/*Closing party submits proof-of-fraud*/}
        \STATE After verifying $2f+1$ on-chain signed announcements from committee $C_A$, $A$ publishes proofs-of-fraud.
        
        \textcolor{magenta}{/*Closing with punishment*/}
        \STATE After the state is included in a block, the smart contract verifies the signature of $VS$ based on the information in the register transaction $TX_r$.
        \STATE Then the smart contract verifies the proofs-of-fraud:
        \begin{enumerate}
            \item[(a)] If the valid proofs about a channel are $x\leq f$, the smart contract follows Protocol~\ref{alg:crosschecking} to decide the \emph{wardens' published state}, $WS_A$. Then the smart contract virtually distributes coins by updating the parties' current balances according to $WS_A$.
            \item[(b)] If the valid proofs about a channel are $x\geq f+1$, smart contracts award all channel balances to the counterparty.
        \end{enumerate}
        \STATE Smart contracts award all the cheating wardens' collateral to party $A$.
        \STATE First $2f+1$ wardens that have no cheating behaviors get an equal fraction of the closing fee.
        \STATE If the virtual channel is not closed before closing the payment channel, the whole channel balance will be given to the payment channel counterparty.
    \end{algorithmic}
\end{algorithm}

\begin{algorithm}[!h]
    \caption{Smart contract cross-checking}
    \label{alg:crosschecking}

    \textbf{Sending transaction}
    \begin{algorithmic}[1]
        \REQUIRE Alice-Ingrid smart contract $SC_A$, Ingrid-Bob smart contract $SC_B$, wardens' published state $WS_A$.
        \ENSURE Send notification and query transaction.

        \textcolor{magenta}{/*Decide state*/}
        \STATE $SC_A$ stores the state with the highest sequence number as $WS_A$. If there is more than one, $SC_A$ selects the first one it receives.

        \textcolor{magenta}{/*Exchange cross-checking transaction*/}
        \STATE $SC_A$ sends the notification and query transaction $TX_{qA}$ containning $WS_A$ to $SC_B$.  

        \STATE $SC_A$ waits for the transaction $TX_{qB}$ from $SC_B$. 
    \end{algorithmic}

    \textbf{Receiving transaction}
    \begin{algorithmic}[1]
    
    \REQUIRE Alice-Ingrid smart contract $SC_A$, Ingrid-Bob smart contract $SC_B$, wardens' published state $WS_A$.
    \ENSURE $SC_A$ decides on the closing state.
        
         \textcolor{magenta}{/*Receive transaction from the other contract*/}
         
         \STATE $SC_A$ receives transaction $TX_{qB}$ containing $WS_B$ from $SC_B$.

        \textcolor{magenta}{/*React when state is not decided yet*/}
         
         \STATE If $SC_A$ hasn't deceided $WS_A$ yet, then stores $WS_B$ as $WS_A$ and reply with $NULL$.

         \textcolor{magenta}{/*React when state is decided decided*/}

         \STATE If $WS_A$ has already been decided, then $SC_A$ reacts as follow:

         \begin{enumerate}
             \item[(a)] If $WS_A$ and $WS_B$ have different sequence number. Then $SC_A$ updates $WS_A$ to the one with the higher sequence number.

             \item[(b)] If $WS_A$ and $WS_B$ have the same sequence number but different value. Then $SC_A$ updates $WS_A$ to the state from the \emph{leader contract}.
         \end{enumerate}

     \end{algorithmic}

\end{algorithm}

\textbf{Bob (Alice) is offline or misbehaves.}
The second scenario occurs when one party, such as Bob, is offline. In this case, Ingrid only receives the closing request from Alice. However, Ingrid cannot immediately agree to Alice's request, as she has no knowledge of the virtual channel state and could potentially be deceived if Alice and Bob are colluding. To address this, Ingrid is allowed to unilaterally close the part of \vc involving the offline party first, thereby gaining access to the virtual channel state. Once Ingrid has this information, she can verify Alice's closing requests and act accordingly. In our game-theoretic analysis, we demonstrate that wardens will continue to publish honestly, even if Ingrid is initially unaware of the \vc state.

\textbf{Alice and Bob are both offline or collude.} The final pessimistic situation covers two scenarios: (1) Ingrid wants to close \vc, but both Alice and Bob are offline; (2) Alice and Bob send valid closing requests but with different closing states. The second scenario suggests that Alice and Bob are colluding during the \vc update procedure, generating two different final states to deceive Ingrid. In both cases, Ingrid cannot rely on Alice and Bob to close \vc properly. Therefore, Ingrid must execute Protocol~\ref{alg:uni-AB} sequentially to close both parts of \vc unilaterally. We stress that if Alice and Bob send identical but outdated closing requests, Ingrid incurs no loss by agreeing to them; so this scenario is excluded from our security analysis.

\textbf{Ingrid and Bob (Alice) are offline.}
In the final possible scenario, which is not covered in Protocol~\ref{alg:close-3online}, both Ingrid and Bob (or Alice) are offline during the closing process. In this case, the online party can execute Protocol~\ref{alg:uni-AB} to close their part of \vc unilaterally. Once the other two parties come back online, they can close their portion of \vc either unilaterally or collaboratively.

\subsection{Multi-hop \vc}
\label{sec:multi}

\vc can also be extended to multi-hop scenarios. We illustrate this with example involving four parties: Alice, Bob, Charlie, and Dave, connected through TPCs (see Fig.~\ref{fig:multihop}).

\begin{figure}
    \centering
    \includegraphics[width=\linewidth]{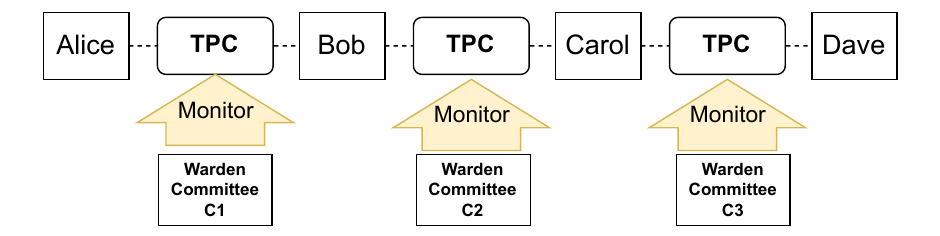}
    \caption{Multi-hop \vc}
    \label{fig:multihop}
\end{figure}

The main difference in a multi-hop scenario is the involvement of additional intermediary parties who are unaware of the \vc states. As with the two-hop \vc, it is crucial to ensure that each committee has enough wardens who know the latest \vc state. The key steps for a multi-hop \vc are as follows: (1) During the opening procedure, all four parties, along with the wardens, must collaborate to generate the \emph{register transaction}. Each party considers the channel open after receiving a quorum certificate from each committee. (2) To successfully update \vc, Alice and Dave must broadcast the updated state to all wardens and wait for a quorum certificate from each committee. (3) \vc can be optimistically closed once all four parties reach an agreement. In a pessimistic closing scenario, each smart contract should wait for at least $2f+1$ wardens to publish their states from the corresponding payment channel committee and cross-check with other contracts to determine the final closing state.

\section{Byzantine Security Analysis}
\label{sec:BAanalysis}

We prove \vc is secure under the Byzantine model when at least $t=2f+1$ wardens are honest in each committee $C_A$ and $C_B$.
The omitted proofs can be found in Appendix~\ref{app:BAvalysis}.
The security analysis of multi-hop \vc is detailed in Appendix~\ref{subsec:multi-byz}.

\begin{restatable}{theorem}{ThmSecurity}
\label{thm:security}
    \vc achieves \textbf{balance security} for honest parties under asynchrony, assuming at most $f$ Byzantine wardens in each committee.
\end{restatable}

\begin{restatable}{theorem}{ThmLiveness}
\label{thm:liveness}
\vc achieves \textbf{liveness} for honest parties under asynchrony, assuming at most $f$ Byzantine wardens in each committee.
\end{restatable}

\section{Game-theoretic Security Analysis}
\label{sec:GTanalysis}

We prove \vc is game-theoretically secure (Def.~\ref{def:Game-theoretic security}) by showing that no party can increase their utility by deviating from the protocol specification during the open, update, or close phases. Detailed proofs are provided in Appendix~\ref{sec:gameproof}, with the multi-hop analysis outlined in Appendix~\ref{subsec:multi-game}. Next, we focus on modeling the protocol as a game.

Since \vc participants act sequentially during closing, we model the process as an \emph{Extensive Form Game (EFG)}~\cite{osborne1994course}. In EFG, all possible protocol executions are represented in a game tree, where nodes denote decision points for players, branches represent potential actions, and leaves indicate utility outcomes based on the chosen strategies.

While \vc comprises three phases, the update phase is effectively reflected in the closing process. Therefore, we model the opening and closing procedures as an EFG to prove the scheme's security. As the opening phase is straightforward, we present the corresponding theorem below but refer the reader to Appendix~\ref{sec:gameproof} for the complete modeling and proof. We now focus on modeling the closing game.


\begin{restatable}{theorem}{Thmopen}
\label{thm:open}
    Rational main parties will open \vc correctly.
\end{restatable}

We assume the whole \vc closing process is initiated by Alice and Bob. According to Protocol~\ref{alg:close-3online}, the closing of Alice-Ingrid's part and Ingrid-Bob's part are symmetric; for simplicity, we only show the game involving Bob and Ingrid. The possible strategies for these two parties are defined in Table~\ref{tab: action-bob} and~\ref{tab: action-ingrid}. Depending on whether Ingrid knows the \vc state, there could be two different game trees depicted in Fig.~\ref{fig:bclose} and~\ref{fig:bclosep}. The blue dotted circle in Fig.~\ref{fig:bclose} represents that the parties inside the circle share the same information set~\cite{osborne1994course}. Specifically, if Bob sends a closing request while Alice is offline, Ingrid cannot determine whether the request is correct, which makes the game an EFG with imperfect information. Conversely, if Alice also sends a closing request or Ingrid gains knowledge of the virtual channel through unilateral closing, the closing game changes to an EFG with perfect information, as shown in Fig.~\ref{fig:bclosep}. \emph{Subgame 1} represents the unilateral closing game, involving the unilaterally closing party and the wardens. This game also reflects the update protocol, where wardens decide whether to sign honestly. Additionally, the cross-checking scheme ensures that both parts of \vc can only be unilaterally closed with the same state, guaranteeing Ingrid’s security and incentivizing rational parties to follow the update protocol. The definition for the subgame can be found in Appendix~\ref{sec:subgame}. We identify the SPNE of the closing game, thereby proving its security, as concluded in Theorem~\ref{thm:final}.

\begin{table}[!h]
\normalsize
\centering
\caption{Strategy for Bob ($B$)}
\label{tab: action-bob}
\resizebox{\columnwidth}{!}{%
\begin{tabular}{|c|p{\columnwidth}|}
\hline
$Old$ & Send a dishonest collaborative closing request with the old virtual channel state \\ \hline
$New$ & Send an honest collaborative closing request with the latest virtual channel state \\ \hline
$Uni$ & Unilaterally close the part of \vc with Ingrid \\ \hline
\end{tabular}%
}
\end{table}

\begin{table}[!h]
\normalsize
\centering
\caption{Strategy for Ingrid ($I$)}
\label{tab: action-ingrid}
\resizebox{\columnwidth}{!}{%
\begin{tabular}{|c|p{\columnwidth}|}
\hline
$Agree$ & Agree and sign the closing request \\ \hline
$Disagree$ & Disagree and go for unilateral closing  \\ \hline
$Ignore$ & Ignore the closing request or go offline \\ \hline
\end{tabular}%
}
\end{table}

\begin{figure}[!h]
    \centering
    \includegraphics[width=0.6\linewidth]{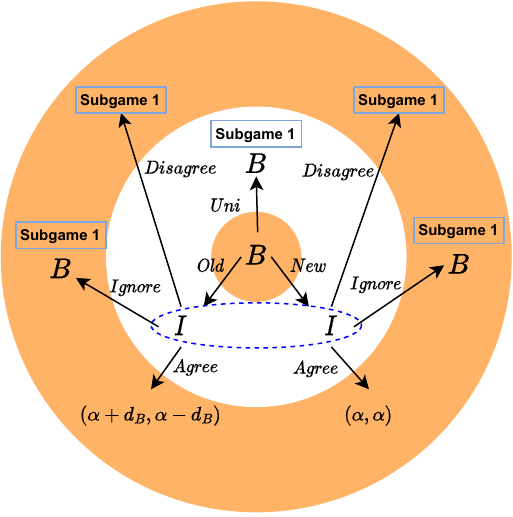}
    \caption{Closing game started by Bob. The utility is presented in the form of (Bob, Ingrid), $\alpha$ represents the incentive for closing \vc, and $d_B$ represents Bob's profit for close \vc incorrectly. The circle with blue dotted line denotes two parties share the same information set.}
    \label{fig:bclose}
\end{figure}

\begin{figure}[!h]
    \centering
    \includegraphics[width=0.6\linewidth]{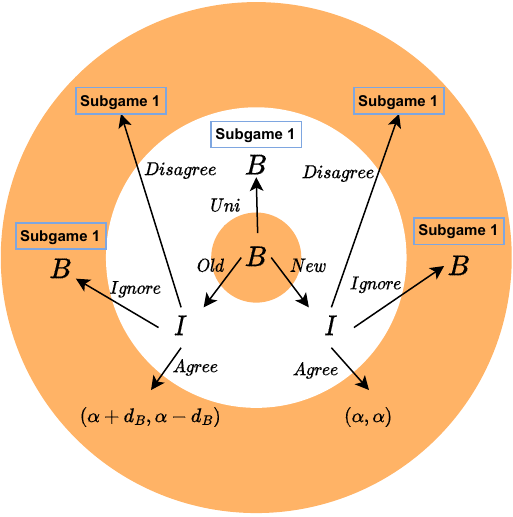}
    \caption{Closing game started by Bob; Ingrid gains knowledge about \vc, distinguishing Bob's requests.}
    \label{fig:bclosep}
\end{figure}

\begin{restatable}{theorem}{Thmgamesec}
\label{thm:final}
\vc closing protocol is game-theoretic secure.
\end{restatable}

\begin{table*}[!t]
\resizebox{\textwidth}{!}{%
\begin{tabular}{|l|c|ccc|ccc|ccc|}
\hline
\multirow{2}{*}{} &
  \multirow{2}{*}{On-chain transaction} &
  \multicolumn{3}{c|}{off/on-chain message} &
  \multicolumn{3}{c|}{\vc Cost} &
  \multicolumn{3}{c|}{Perun Cost} \\ \cline{3-11} 
 &
   &
  \multicolumn{1}{c|}{Alice (Bob)} &
  \multicolumn{1}{c|}{Ingrid} &
  Warden (10) &
  \multicolumn{1}{c|}{Gas} &
  \multicolumn{1}{c|}{ETH} &
  USD &
  \multicolumn{1}{c|}{Gas} &
  \multicolumn{1}{c|}{ETH} &
  USD \\ \hline
Deploy \& Open PC &
  2+10 &
  \multicolumn{1}{c|}{1+10 / 1} &
  \multicolumn{1}{c|}{1+10 / 1} &
  $14\leq m \leq 20$ / 10 &
  \multicolumn{1}{c|}{4444861} &
  \multicolumn{1}{c|}{0.0089} &
  14.83 &
  \multicolumn{1}{c|}{2819448} &
  \multicolumn{1}{c|}{0.0057} &
  9.54 \\ \hline
Update PC &
  0 &
  \multicolumn{1}{c|}{1+10 / 0} &
  \multicolumn{1}{c|}{1+10 / 0} &
  $14\leq m \leq 20$ / 0 &
  \multicolumn{1}{c|}{0} &
  \multicolumn{1}{c|}{0} &
  0 &
  \multicolumn{1}{c|}{0} &
  \multicolumn{1}{c|}{0} &
  0 \\ \hline
Open VC &
  0 &
  \multicolumn{1}{c|}{3+30 / 0} &
  \multicolumn{1}{c|}{2+40 / 0} &
  $35\leq m \leq 50$ / 0 &
  \multicolumn{1}{c|}{0} &
  \multicolumn{1}{c|}{0} &
  0 &
  \multicolumn{1}{c|}{0} &
  \multicolumn{1}{c|}{0} &
  0 \\ \hline
Update VC &
  0 &
  \multicolumn{1}{c|}{1+20 / 0} &
  \multicolumn{1}{c|}{0 / 0} &
  $14\leq m \leq 20$ / 0 &
  \multicolumn{1}{c|}{0} &
  \multicolumn{1}{c|}{0} &
  0 &
  \multicolumn{1}{c|}{0} &
  \multicolumn{1}{c|}{0} &
  0 \\ \hline
Optimistic VC close &
  0 &
  \multicolumn{1}{c|}{1+10 / 0} &
  \multicolumn{1}{c|}{1+10 / 0} &
  $14\leq m \leq 20$ / 0 &
  \multicolumn{1}{c|}{0} &
  \multicolumn{1}{c|}{0} &
  0 &
  \multicolumn{1}{c|}{0} &
  \multicolumn{1}{c|}{0} &
  0 \\ \hline
Optimistic PC close &
  2 &
  \multicolumn{1}{c|}{1 / 1} &
  \multicolumn{1}{c|}{1 / 1} &
  0 / 0 &
  \multicolumn{1}{c|}{252760} &
  \multicolumn{1}{c|}{0.0005} &
  0.84 &
  \multicolumn{1}{c|}{147788} &
  \multicolumn{1}{c|}{0.0003} &
  0.49 \\ \hline
\begin{tabular}[c]{@{}l@{}}Pessimistic VC close\\ \\ by Alice\end{tabular} &
  $2+7\leq m \leq 2+10$ &
  \multicolumn{1}{c|}{0 / 2} &
  \multicolumn{1}{c|}{0 / 0} &
  0 / $7\leq m \leq 10$ &
  \multicolumn{1}{c|}{1217307} &
  \multicolumn{1}{c|}{0.0024} &
  4.06 &
  \multicolumn{1}{c|}{418318} &
  \multicolumn{1}{c|}{0.0008} &
  1.40 \\ \hline
\begin{tabular}[c]{@{}l@{}}Pessimistic PC close\\ \\ by Alice\end{tabular} &
  $1+7\leq m \leq 1+10$ &
  \multicolumn{1}{c|}{0 / 1} &
  \multicolumn{1}{c|}{0 / 0} &
  0 / $7\leq m \leq 10$ &
  \multicolumn{1}{c|}{862459} &
  \multicolumn{1}{c|}{0.0017} &
  2.88 &
  \multicolumn{1}{c|}{275049} &
  \multicolumn{1}{c|}{0.0006} &
  0.92 \\ \hline
\end{tabular}%
}
\caption{Gas and communication cost for different procedures in \vc and Perun~\cite{dziembowski2019perun}. Warden refers to a single payment channel warden committee. PC denotes a payment channel and VC denotes a virtual channel. The column ``on-chain'' transaction counts transactions sent by the main parties (before ``+'') and by wardens (after ``+''). The column ``off/on-chain message'' counts the off-chain messages among the main parties (before ``+'') and to wardens (after ``+''), as well as the on-chain messages.}
\label{tab:gas fee}
\end{table*}

\section{Implementation and Evaluation}
\label{sec:eva}

In this section, we present a proof-of-concept implementation of \vc. The smart contract for the channel is written in Solidity V0.8.21 and deployed using Truffle V5.11.4 on a private Ethereum testing blockchain set up with Ganache V7.9.1. The source code is publicly available at https://github.com/BartWaaang/Thunderdome. We evaluate the gas fee cost of smart contracts on Ethereum as our primary metric. Gas fees are paid by blockchain users to miners as transaction fees, which depend on both the amount of data and the complexity of operations in the transaction. In our implementation, we set the per unit gas price to $2 \times 10^{-9}$ ETH, a typical setting. Additionally, we set the exchange rate between ETH and USD at $1668:1$, based on the latest rate at the time of writing. Under these conditions, the total cost to deploy the payment channel contract supporting our \vc protocol is approximately 0.006 ETH (3,174,554 gas, around 10.59 USD). We evaluate the cost for each procedure using an evaluation script written in JavaScript, which generates testnet accounts for the main parties and wardens, simulating their actions. The gas fee costs for various procedures in our protocol are summarized in Table~\ref{tab:gas fee}.

\begin{figure}[!h]
    \centering
    \includegraphics[width=\linewidth]{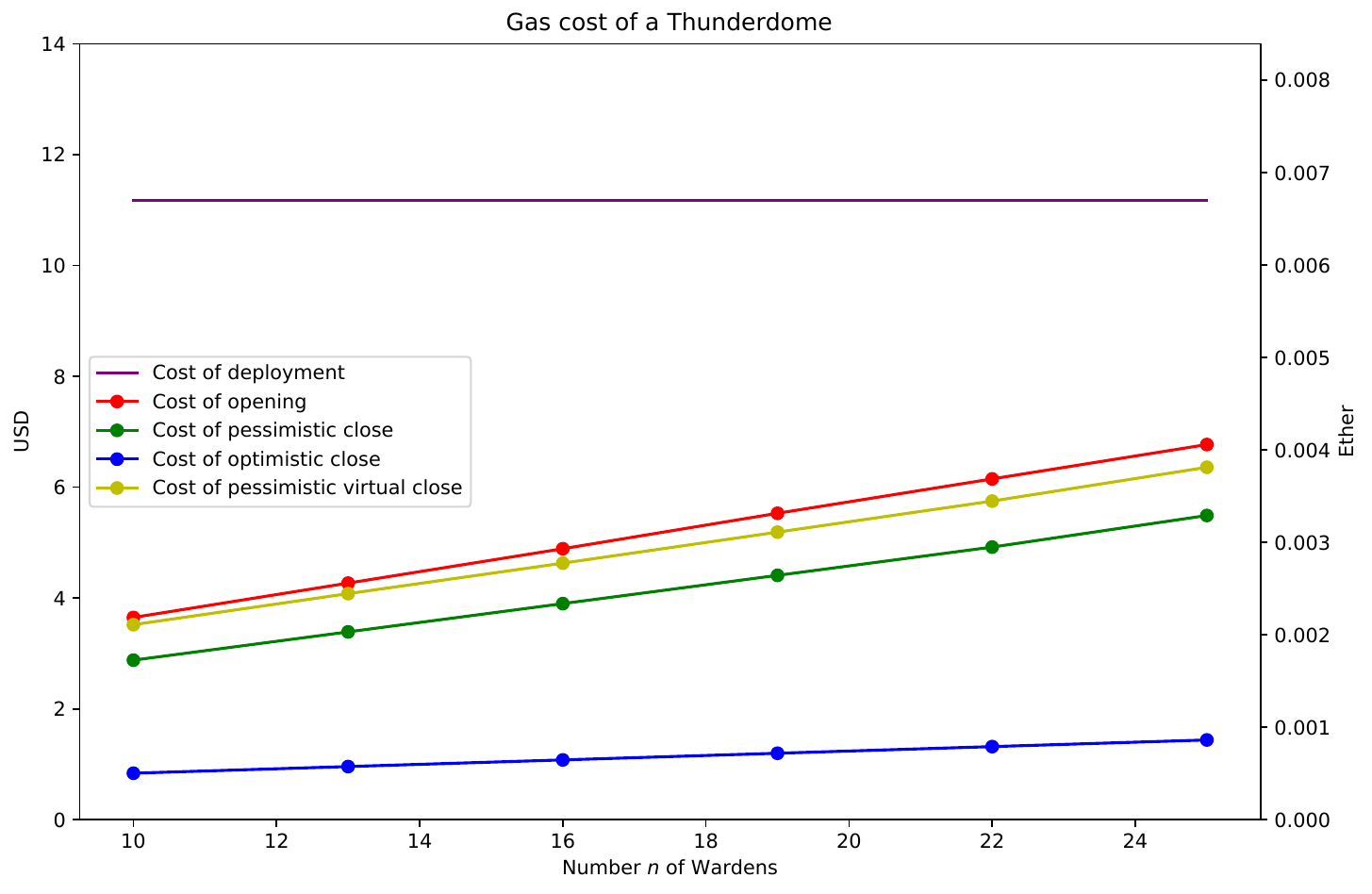}
    \caption{Gas fee changes in \vc 's procedures as the number of wardens per committee grows from 10 to 25.}
    \label{fig:gaschange}
\end{figure}

We compare our gas fee costs with Perun, as shown in Table~\ref{tab:gas fee}. To deploy and open an underlying payment channel (PC), Alice (Bob) and Ingrid, along with 10 wardens, send funding transactions on-chain. The total cost for \vc amounts to 0.0089 ETH, which is 1.5 times the cost of Perun, primarily due to the additional funding transactions required from the wardens. Given the need for extra funding transactions and the removal of timelocks, this increase is reasonable. Regarding off-chain communication, \vc requires quorum certificates from committees to open the channel. Additionally, updating the payment channel for \vc's "virtual lock" involves two rounds of broadcasting, leading to higher message complexity compared to updating a standard payment channel. To close the payment channel optimistically, Alice and Ingrid must publish the closing state agreed upon by both parties on-chain, costing 0.0005 ETH. The pessimistic closing, where at least 7 wardens publish the state on-chain after one party (e.g., Alice) submits a closing request, costs approximately 0.0017 ETH. The message cost for the optimistic closing of the virtual channel is the same as for a payment channel update. In pessimistic situations, closing the virtual channel requires the closing party to submit the register transaction to the smart contract and wardens to publish \vc states on-chain, costing 0.0024 ETH — 3 times the cost of Perun. This additional closing cost arises because all wardens must publish on-chain transactions for the closing state, similar to the \vc opening process. However, given that pessimistic scenarios like censorship are infrequent and the difference in optimistic payment channel closing is only around 0.4 USD, \vc remains comparable to Perun in most cases. We also evaluate how gas fees change as the size of each warden committee increases from 10 to 25. Naturally, costs for all procedures rise with more wardens, as more participants need to interact with the smart contract during \vc execution. The results are illustrated in Fig.~\ref{fig:gaschange}.

\section{Extensions}
\noindent\textbf{Privacy-preserving protocol.} In \vc, the privacy of virtual channel states is not protected from the wardens. It is not straightforward to design a privacy-preserving protocol that ensures both: a) Ingrid can recover the state, and b) the wardens cannot retrieve the state, even in collusion with Ingrid, while still being able to verify the validity of the private state. A compromise is to enable state recovery and privacy against wardens, assuming Ingrid does not collude with them. However, even in this case, the repetitive nature of updates poses challenges. Trapdoor one-way functions, though promising candidates, lack IND-CPA security. On the other hand, asymmetric encryption schemes offering such protection require a new random number for each encrypted value~\cite{bellare2009hedged}. In our protocol, this can be implemented by Alice, Bob, and Ingrid sharing a vector of random values during the creation of the register transaction in \vc open. However, this solution has two main drawbacks: a) the number of states that can be updated in the virtual channel is limited by the length of this vector. Optimistically, Ingrid may occasionally renew it, but this depends on Ingrid being online and cooperative; b) the virtual channel parties must provide the wardens with corresponding evidence (e.g., ZK proof) at each update to certify that the correct random number was used, adding communication and computational overhead to \vc's update process. Other cryptographic tools, such as Verifiable Secret Sharing~\cite{stadler1996publicly} or threshold signatures~\cite{shoup2000practical}, could be leveraged to enhance privacy in \vc transactions. For example, one could explore the possibility of allowing Ingrid and $f+1$ wardens from her committee to recover the states. This scheme may be secure under the assumption of $2f+1$ honest wardens, but it remains an open question whether suitable incentive mechanisms can ensure security in the rational model.

\vspace{0.2cm}\noindent\textbf{Scaling the multi-hop protocol.} An interesting question is whether it is feasible to scale the multi-hop \vc by involving fewer wardens. We argue that this is possible by defining the employed wardens and the closing rules during the opening phase. Assuming that only $n' < n$ wardens are selected per committee, two necessary compromises arise: (1) \emph{Selection algorithm}: Among $n' = 3f' + 1$ wardens, there should still be at most $f'$ Byzantine wardens; (2) \emph{Extra deposit}: If the multi-hop \vc channel balance is $v$, then each warden must lock collateral of at least $\frac{v}{f'} > \frac{v}{f}$. However, the validity rule (e.g., requiring a quorum certificate per committee during the update procedure) remains unchanged for two key reasons. First, the security model assumes that the two end parties of the multi-hop \vc could collude, necessitating warden committees to ensure there are no two conflicting valid states. Second, intermediary parties without knowledge of the virtual channel state rely on the wardens in the local committee to publish the correct state.  

\vspace{0.2cm}\noindent\textbf{Extend to state channel.} 
\vc is introduced as a virtual payment channel but can also be extended to function as a virtual state channel. Although numerous state channel primitives exist in the literature~\cite{dziembowski2018general,miller2019sprites,mccorry2019pisa,mccorry2020you} that generalize payment channels to more complex applications, these constructions cannot be directly applied to \vc due to their synchronous nature, which relies on timelocks. However, \vc can utilize two Brick state channels~\cite{avarikioti2019brick2} as the underlying layer, provided that each state can be mapped to a monetary valuation to ensure security in the game-theoretic model. This assumption, as proposed in \cite{avarikioti2019brick2}, is essential for converting the payment channel primitive into a rationally secure state channel, which \vc also adopts.

\section*{Acknowledgement}

This work was supported by the Vienna Science and Technology Fund (WWTF) through the project 10.47379/ICT22045;

\bibliographystyle{plain}
\bibliography{reference}

\appendix

\section{Byzantine Model Security Proof}\label{app:BAvalysis}

To start with, we show that the cross-checking scheme maintains security in the following lemma.

\begin{lemma}
\label{lem:secupdate}
     The cross-checking scheme guarantees \textbf{balance security} for Ingrid when Alice and Bob attempt to unilaterally close the channel at the same time.
\end{lemma}

\begin{proof}
    Assume Alice and Bob collude to violate \textbf{balance security} for Ingrid. Without loss of generality, there are two possible cases: (1) Alice closes one part of \vc with a state having a lower sequence number, while Bob closes with a different state with a higher sequence number; (2) Alice and Bob close with states that have different values but the same sequence number.

    For the first case, assume Alice initiate unilateral closing first. After receiving $2f+1$ signatures from the wardens, Alice's contract will send the state notification to Bob's contract. According to Protocol~\ref{alg:crosschecking}, since Bob's contract has not yet received enough signatures from its wardens, it will store the state received from Alice's contract and reply with a Null. Subsequently, regardless of the state later published to Bob's contract, it will only close with the state received from Alice's contract. Conversely, if Bob initiates the closing first, then both parts of \vc will close with Bob's state. In the event that Alice and Bob initiate unilateral closings simultaneously, both contracts will receive two different states—one from Alice and one from Bob. By comparing the sequence numbers, both contracts will decide on Bob's state, as it has the higher sequence number. Thus, the first case cannot occur.

    For the second case, similar reasoning applies. Both parts of \vc will close with the state from the party who initiates the unilateral closing first. If Alice and Bob start unilateral closings simultaneously, the two states known to both contracts will have the same sequence number. To resolve this tie, the "leader contract" scheme is employed according to Protocol~\ref{alg:crosschecking}. Since the leader contract requires agreement from all three main parties, the existence of at least one honest party ensures that there is only one leader contract. If Alice's contract serves as the leader contract, both parts of \vc will close with the state from Alice's contract. Consequently, the second case also cannot occur.

    Conclusively, the two parts of \vc can only be closed with the same state, ensuring that \textbf{balance security} is always guaranteed for Ingrid.
\end{proof}

Based on Lemma~\ref{lem:secupdate}, we can prove Theorem~\ref{thm:security} and~\ref{thm:liveness}.

\ThmSecurity*

\begin{proof}
    To prove the balance security of \vc, we aim to show that no honest participants lose coins when the \vc is closed. First, \vc is set up correctly according to the protocol, as honest parties will follow the protocol and prevent Byzantine participants from improperly updating the payment channel. Therefore, we focus on the update and close procedures. As shown in Lemma~\ref{lem:secupdate}, security is maintained even when parties deviate from the update protocol. In the optimistic closing scenario, where all parties are honest and responsive, once all three parties agree to close \vc, the two underlying TPCs are updated according to the latest virtual channel state. As long as all parties behave rationally, the optimistic closing of \vc is balance secure, as no rational party would agree to close the channel incorrectly.

    For pessimistic situations, where some parties are offline or send dishonest closing requests, we begin by considering cases where only one main party is Byzantine. According to Protocol~\ref{alg:close-3online}, the honest online party will invoke Protocol~\ref{alg:uni-AB} f the payment channel's counterparty is offline or Byzantine. We prove balance security in this scenario by contradiction. Let $s_i$ denote the closing state of the virtual channel. Suppose there is a later state $s_k$ such that $k > i$. Thus, $s_i$ is not the latest valid virtual channel state agreed upon by the end parties. If one part of the virtual channel closes with $s_i$ (e.g., the Alice-Ingrid part), at least $t = 2f+1$ wardens should have proposed closing announcements, and the maximum sequence number would be $i$, per Protocol~\ref{alg:uni-AB}. Under the threat model, at most $f$ wardens in $C_A$ can be Byzantine. Therefore, at least $f+1$ honest wardens should not have received state $s_k$. However, a valid virtual channel state must be verified by $t = 2f+1$ wardens, as specified by Protocol~\ref{alg:vupdate}. In this case, the number of wardens required would exceed the size of committee $C_A$, leading to a contradiction.

    Moreover, when only Ingrid and Alice are online, and Alice is Byzantine, Ingrid can distinguish Alice's dishonest closing request by unilaterally closing the channel with Bob first. Thus, this analysis also accounts for the intermediary party having no prior knowledge of the \vc state.

    Next, we consider the pessimistic case where two main parties are Byzantine; say Alice and Ingrid. If Bob is online, he can invoke Protocol~\ref{alg:uni-AB} to unilaterally close the Ingrid-Bob part of \vc without incurring any loss. Furthermore, according to the analysis of Protocol~\ref{alg:uni-AB}, Alice and Ingrid cannot close the Ingrid-Bob part incorrectly, even if Bob is offline. Similarly, if Alice and Bob are Byzantine, and Ingrid is online, Ingrid can securely close \vc by running Protocol~\ref{alg:uni-AB}to sequentially close the Alice-Ingrid and Ingrid-Bob parts. Alice and Bob also cannot close these parts differently or incorrectly if Ingrid is offline.
\end{proof} 

\ThmLiveness*

\begin{proof}
    We now demonstrate that every valid operation, whether an \textit{update} or \textit{close}, is either committed or invalidated.

Suppose the operation is a valid update that was never committed. If the operation is valid and at least one party in the channel is honest, the wardens will eventually receive the state update. However, since the update was never committed, there must be at least $f + 1$ wardens who did not sign the state. Given that at most $f$ wardens in each committee can be Byzantine, at least one honest warden must have refrained from signing the valid update state. According to Protocol~\ref{alg:vupdate}, an honest warden performs specific verifications, and if those verifications are satisfied, the warden signs the new state. The only reason an honest warden would not sign is if the state were invalid, which contradicts our assumption. If both end parties of the virtual channel are Byzantine, the update cannot occur; however, the honest Ingrid can unilaterally close the channel to adjust or retrieve the balance.
    
    Next, suppose the operation is a valid close. In the optimistic scenario, where all parties are online, the payment channel involving the honest party will be correctly updated or closed using the virtual channel state. Due to the liveness property of the TPC, the virtual channel state will become committed. In the case where some parties are offline, any online party can unilaterally close \vc with the assistance of wardens and smart contracts.

Thus, liveness is ensured under the Byzantine model. 
\end{proof}

\section{Game-theoretic Model}
\label{sec:gamemodel}
In the game-theoretic setting, our goals mirror the security properties in the Byzantine model. However, to prove these properties we employ different tools as now we consider rational agents that try to maximize their utility. To this end, we model our protocol as an Extensive Form Game (EFG) with Perfect Information.

\begin{definition}[Extensive Form Game-EFG]

And Extensive Form Game (EFG) is a tuple $\mathcal{G}=(N,H,P,u)$, where set $N$ represents the game player, the set $H$ captures EFG game history, $T \subseteq H$ is the set of terminal histories, $P$ denotes the next player function, and $u$ is the utility function. The following properties are satisfied.

(A) The set $H$ of histories is a set of sequence actions with
    \begin{enumerate}
        \item $\emptyset \in H$;
        \item if the action sequence $(a_k)^K_{k=1}\in H$ and $L<K$, then also $(a_k)^L_{k=1} \in H$;
        \item an action sequence is terminal $(a_k)^K_{k=1}\in T$, if there is no further action $a_{K+1}$ that $(a_k)^{K+1}_{k=1}\in H$.
    \end{enumerate}

(B) The next player function $P$
    \begin{enumerate}
        \item assigns the next player $p \in N$ to every non-terminal history $(a_k)_{k=1}^K \in H\setminus T$;
        \item after a non-terminal history $h$, it is player $P(h)$'s turn to choose an action from the set $A(h)=\{a:(h,a)\in H\}$.
    \end{enumerate}
    
A player $p$'s strategy is a function $\sigma_p$ mapping every history $h \in H$ with $P(h)=p$ to an action from $A(h)$. Formally,$$\sigma_p:\{h \in H:P(h)=p\} \rightarrow\{a:(h,a)\in H, \forall h\in H \},$$
such that $\sigma_p(h)\in A(h)$.

\end{definition}
The subgame of an EFG is a subtree determined by a certain history (i.e., whose root note is the last history node) and is formalized by the following definition~\cite{rain2022towards}:

\begin{definition}[EFG subgame]
    The subgame of an EFG $\mathcal{\varphi} = (N,H,P,u)$ associated to history $h \in H$ is the EFG $\mathcal{\varphi}(h) = (N,H_{|h},P_{|h},u_{|h})$ defined as follows: $H_{|h} := {h' | (h,h') \in H}$, $P_{|h}(h') := P(h,h')$, and $u_{|h}(h') := u(h,h')$.
\end{definition}

The core concept of our proof methodology is to demonstrate that utility-maximizing players will consistently choose to adhere to the protocol specification at every step of the protocol. Our focus primarily lies in the closing process of \vc. We achieve this by leveraging the concept of Subgame Perfect Nash Equilibrium (SPNE) (Definition~\ref{def:spne}) \cite{osborne1994course}. In particular, in Section~\ref{sec:GTanalysis}, we show that the strategy profile encompassing the `correct protocol execution' of \vc is an SPNE of our game, using a technique known as backward induction.

\begin{definition}[Subgame Perfect Nash Equilibrium (SPNE)]
\label{def:spne}
A subgame perfect equilibrium strategy is a joint strategy $\sigma = (\sigma_1, ...,\sigma_n) \in S$, s.t. $\sigma_{|h} = (\sigma_{1|h}, ...,\sigma_{n|h})$ is a Nash Equilibrium of the subgame $\mathcal{\varphi}(h)$, for every $h \in H$. The strategies $\sigma_{i|h}$ are functions that map every $h' \in H_{|h}$ with $P_{|h}(h') = i$ to an action from $A_{|h}(h')$.
\end{definition}

Using the definitions above, we establish a security property for the closing process of \vc in the game-theoretic model, as outlined in Definition~\ref{def:Game-theoretic security}. It is important to note that the security of the opening process of \vc is straightforward, assuming that the parties are willing to open \vc, and the channel will be correctly opened according to the protocol specification.

\begin{definition}[Game-theoretic security]
\label{def:Game-theoretic security}
    A closing protocol is \emph{game-theoretic secure} if each main party $n\in \{A,I,B\}$'s utility of the closing game $\mathcal{G}$'s history $h^*$, which is composed by the SPNE joint strategy $\sigma^*$, is no less than $\alpha-\epsilon$, $u_{n}(h^*) \geq \alpha-\epsilon$.  
\end{definition}
Here, $\alpha$ represents the regular profit parties receive from closing \vc, ensuring that they are incentivized to close the virtual channel, and $\epsilon$ denotes the cost associated with unilateral closing, required when some parties are offline.

\section{Unilateral Closing Game}
\label{sec:subgame}

The unilateral closing game involves two types of participants: the unilateral closing party and the wardens, with their strategies defined in Table~\ref{tab: action-p} and Table~\ref{tab: action-warden}. Each warden has two strategies: $P_l$ and $P_o$. Assuming Bob is the current party unilaterally closing, and the Alice-Ingrid part has already been closed unilaterally, $P_l$ represents two possible actions: \textit{publish the latest state known to the warden} or \textit{publish the same state used for Alice-Ingrid's closing}. Conversely, $P_o$ represents publishing an old or inconsistent state. For simplicity, we denote these two strategies as "publish the honest state" and "publish the dishonest state." For the closing party, the only strategy is $PF$, which involves deciding how many proofs-of-fraud to publish. This game can also represent the update game, where $P_l$ corresponds to signing the state honestly, $P_o$ to signing two different states, and $PF$ to providing proofs-of-fraud. To avoid redundancy, we focus on the proof for the closing game.

\begin{table}[!h]
\normalsize
\centering
\caption{Strategy for Closing party ($P$)}
\label{tab: action-p}
\resizebox{\columnwidth}{!}{%
\begin{tabular}{|c|p{\columnwidth}|}
\hline
$PF$ & Number of valid proofs-of-fraud she publishes on-chain \\ \hline
\end{tabular}%
}
\end{table}

\begin{table}[!h]
\normalsize
\centering
\caption{Strategy for warden $i$ ($W_i$)}
\label{tab: action-warden}
\resizebox{\columnwidth}{!}{%
\begin{tabular}{|c|p{\columnwidth}|}
\hline
$P_l$ & Publish the honest state $\backslash$ Sign honestly  \\ \hline
$P_o$ & Publish the dishonest state $\backslash$ Sign dishonestly  \\ \hline
\end{tabular}%
}
\end{table}

During the unilateral closing process, multiple wardens participate in the game. For simplicity, we assume that each payment channel committee consists of $3f+1=10$ wardens. According to Protocol~\ref{alg:uni-AB}, at least 7 of these wardens from the closing party's warden committee are required to publish states onchain. To further simplify the analysis, we divide these 10 wardens into three groups, denoted as $\{W_1, W_2, W_3\}$. Specifically, $W_1$ contains 3 wardens who are unaware of the latest virtual channel state, while $W_2$ and $W_3$ consist of 4 and 3 wardens, respectively, who know the latest state. Since all wardens are treated identically in this analysis, we focus only on the total number of wardens adopting a particular strategy. Additionally, we introduce the following assumption for wardens' strategies:
\begin{assumption}
\label{asm:simultanious}
    The wardens act simultaneously with respect to their interaction with the blockchain.
\end{assumption}

Assumption~\ref{asm:simultanious} is reasonable for several reasons. First, wardens are compensated to monitor the state of the blockchain, ensuring that they observe any changes within a specific time frame. Second, any transaction that is to be published on the blockchain must first be submitted to miners, who then include these transactions in the block they are mining. Each block, eventually added to the blockchain, contains multiple transactions. Additionally, our protocol includes an incentive mechanism, provided through other unidirectional channels~\cite{gudgeon2020sok}, which does not affect the \vc protocol. The first set of wardens to act will receive the incentive coins, meaning wardens must decide on transaction publishing simultaneously, without waiting for others to act before making their own decisions.

Consequently, our game analysis treats these 10 wardens as a single party. For the order of published states, we consider the worst-case scenario, where the wardens' transactions are published on-chain in the order of $W_1$ to $W_3$, w.l.o.g. Since all three wardens in $W_1$ experience latency, even if they choose $P_l$ as their strategy, they will still publish outdated states, and the closing party cannot provide proofs-of-fraud against them. As a result, the wardens' strategy combinations can be represented as $\{3P_l, aP_l + bP_o, /\}$. Here, $3P_l$ denotes the worst case in which the three wardens from $W_1$ honestly publish outdated states. The term $aP_l + bP_o$ represents the wardens in the second group ($W_2$), where $a$ wardens publish the latest state, and $b$ wardens may have been bribed by the closing party to publish outdated states on-chain. Since we only consider the worst-case scenario where only 7 wardens with collateral publish states, the strategies of the wardens in the third group ($W_3$) are disregarded in our analysis.

After simplifying the wardens' strategy combinations into a single party strategy, denoted as $W$, we can now present the model of the unilateral closing protocol, illustrated in Fig.~\ref{fig:uni}. To express each party's utility function, we use the following notation: $\alpha$ represents the regular profit from closing \vc; $d_A$ denotes the profit for the closing party if a part of \vc is closed with an incorrect state; $k$ represents the incentive fee awarded to the wardens who publish the state; and $c$ denotes a warden's on-chain collateral.

\begin{figure*}
        \centering
    \includegraphics[width = .7\textwidth]{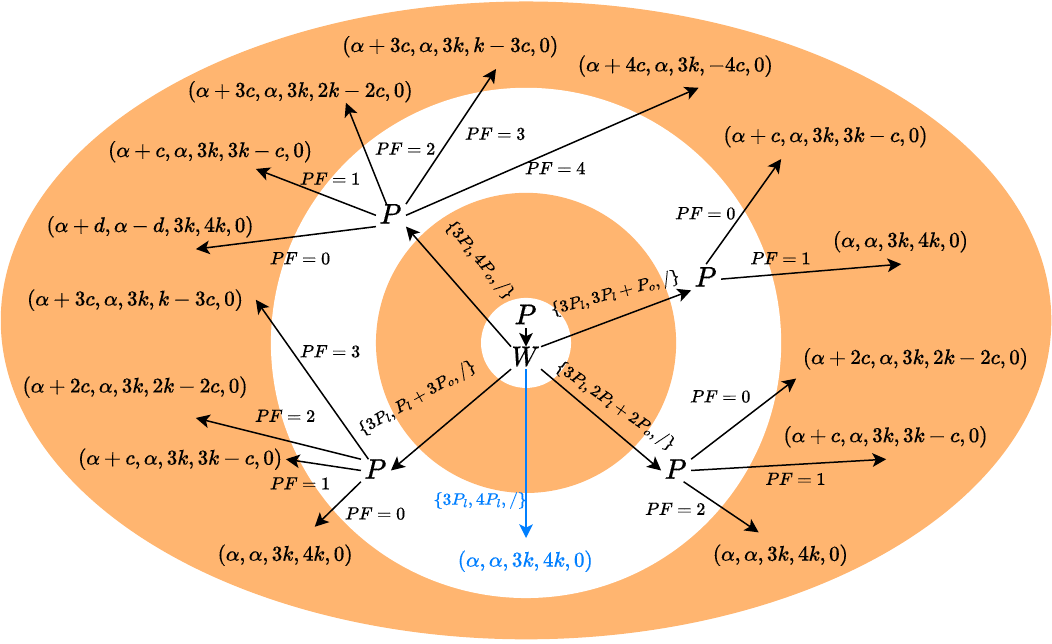}
    \caption{Subgame 1 -- Unilateral closing game started by party $P$, $P \in \{A,B,I\}$. The utility is shown in order of \{$P$, $P$'s counterparty, warden group 1, warden group 2, warden group 3\}. The SPNE of the game is shown with the blue arrow.}
    \label{fig:uni}
\end{figure*}

\section{Game-theoretic Analysis (omitted proofs)}
\label{sec:gameproof}

\begin{figure}
    \centering
    \includegraphics[width=0.7\linewidth]{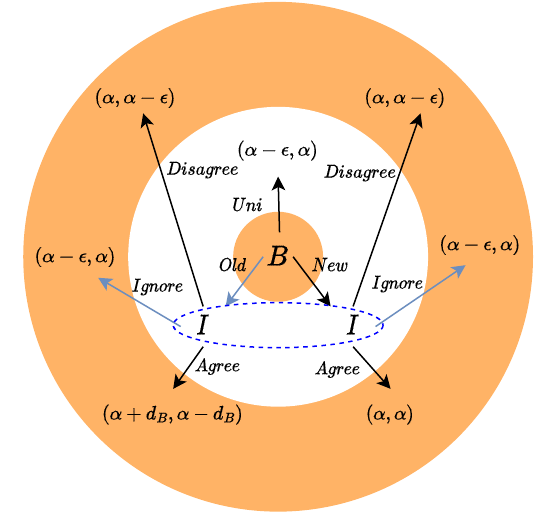}
    \caption{Closing game started by Bob with Ingrid has no knowledge about \vc. $\epsilon$ is a small value to denote the cost for unilateral closing. The SPNE strategy is shown by the blue arrow.}
    \label{fig:bcloseimpE}
\end{figure}

\begin{figure}
    \centering
    \includegraphics[width=0.7\linewidth]{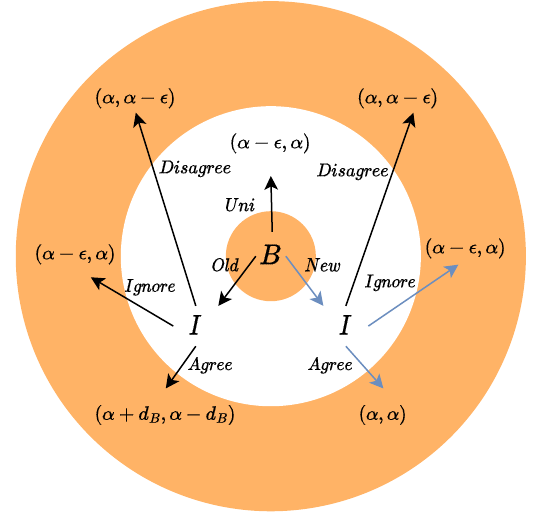}
    \caption{Closing game started by Bob with Ingrid has knowledge about \vc. $\epsilon$ is a small value to denote preference. The SPNE strategy is shown by the blue arrow.}
    \label{fig:bclosepE}
\end{figure}

To prove the game-theoretic security of \vc, we start by establishing the security of the opening procedure, as demonstrated in Theorem~\ref{thm:open}.

\Thmopen*

\begin{proof}
    Assume the current state of both the Alice-Ingrid and Ingrid-Bob payment channels is $(a,b)$. If all three parties agree to open a virtual channel between Alice and Bob with an initial state of $(c,d)$, both channels are updated to $(a-c, b-d)$. A security concern arises if the funding transactions differ. For example, Ingrid could collude with Alice and mislead Bob by updating the Ingrid-Bob channel with an amount $d' \neq d$. If $d' > d$, Ingrid may incur a loss, as Alice would be able to pay Bob more coins in the virtual channel than she has locked in the Alice-Ingrid channel. Conversely, if $d' < d$, Alice could lose funds. Since either Alice or Ingrid would face a financial loss in this scenario—and both are modeled as utility-maximizing agents—such an incorrect opening would not occur.
\end{proof}

After establishing the security of the opening process, the remaining key step in proving the security of \vc under the game-theoretic model is identifying the Subgame Perfect Nash Equilibrium (SPNE) strategy for the closing game. The SPNE of the closing game can be derived by comparing the utilities at the leaf nodes and working backward to the root of the game tree, using backward induction. We begin with an analysis of the unilateral closing subgame. For the unilateral closing game, the following holds:

\begin{lemma}
\label{lem:subgame1}
    As long as the warden collateral $c$ and the \vc channel balance $v$ satisfy $(f+1)c>v$, the SPNE strategy for subgame 1 $\mathcal{G}_{sub1}$ is $\{\sigma_P(\bot)=\bot,\sigma_W(\bot)=P_l\}$
\end{lemma}

\begin{proof}
    
    In the unilateral closing game $\mathcal{G}_U$ shown in Fig.~\ref{fig:uni}, there is only one history, $h_r$, that results in the unilateral closing party's counterparty losing coins. In this history $h_r$, where player $P$ reaches a final utility $u_P(h)=\alpha+d$, the warden group $W$'s strategy is $\sigma^r{W}(\bot)=\{3P_o, 4P_o, /\}$, and the closing party $P$'s strategy is $\sigma^r_P(\{3P_o, 4P_o, /\})$. For this joint strategy $\sigma^r=(\sigma^r_{W_2}, \sigma^r_P)$ to be an SPNE strategy, $P$'s utility in history $h_r$ must dominate other histories, which requires $d>(f+1)c=4c$. Otherwise, $P$'s SPNE strategy in the subgame $\mathcal{\varphi}(\{3P_o, 4P_o, /\})$ becomes $\sigma_P(\{3P_o, 4P_o, /\})= {PF=4}$.

Thus, as long as $(f+1)c > v \geq d$ holds for our \vc protocol, we can ensure that $\sigma^r$ will not be the SPNE strategy. By comparing the second warden group $W_2$'s utility across all other histories, we find that for $W_2$'s strategy $\sigma_{W_2}(\bot)=\{\{aP_l+bP_o\} | a+b=4, a\geq0, b\geq0\}$, where $P_o$ is included, $P$'s SPNE strategy will be $\sigma_P(\sigma_{W_2}(\bot))=\{PF=b\}$. The final utilities for $W_2$ in these histories are $u'=ak-bc \leq 4k$, which is no greater than the final utility of history $h^*=\{\bot,(\{3P_l, 4P_l, /\})\}$, where all wardens publish the latest state on-chain. 
\end{proof} 

We note that Subgame 1 illustrates the strategy and utility of wardens in one payment channel committee. In practice, some of these wardens may also participate in another unilateral closing game initiated by a \vc counterparty due to committee overlap. Although the wardens remain the same, they deposit collateral in different payment channel smart contracts. Therefore, these two games are independent of each other. Additionally, we further prove that the result of Lemma~\ref{lem:subgame1} holds for the closing party $I$, even if $I$ lacks knowledge of the \vc states.

\begin{lemma}
    
\label{lem:distinguish}
    For the unilateral closing game $\mathcal{G}_{I}$ started by Ingrid, who may not know \vc states, the joint SPNE strategy is $\{\sigma_I(\bot)=\bot,\sigma_W(\bot)=P_l\}$, for $p_1>0$, where $p_1$ represents the probability for Ingrid to collude with the channel counterparty.
\end{lemma}

\begin{proof}
As shown in Fig.~\ref{fig:distinguish}, the final utility for wardens depends on the probability $p \in \{p_1, p_2 | p_1 + p_2 = 1\}$ that Ingrid has knowledge of the \vc state. We simplify the original unilateral closing game $\mathcal{G}_I$ by categorizing the warden's strategy into two types. The honest strategy, $\sigma^h_W(\bot)=\{3P_l, 4P_l, /\}$, represents each warden publishing the latest state they know. The dishonest strategy, $\sigma^{dish}_W(\bot)=\{3P_l, aP_l + bP_o, /\}$, represents $b$ wardens publishing outdated states on-chain, where $a + b = 4$ and $0 \leq a < 4$. The warden's expected utility for the dishonest strategy $\sigma^{dish}W(\bot)$ is $u_{dish} = p_1(ak - bc) + 4p_2k = (ap_1 + 4p_2)k - bcp_1$. The expected utility for the honest strategy $\sigma^h_W(\bot)$ is $u_{h} = 4k$. Given that $a < 4$ and $p_1 > 0$, it follows that $u_{dish} = (ap_1 + 4p_2)k - bcp_1 < 4k = u_{h}$. Therefore, the honest strategy $\sigma^h_W(\bot)=\{3P_l, 4P_l, /\}$ yields a higher utility, proving that Lemma~\ref{lem:subgame1} still holds for the closing party without knowledge of the \vc state. 
\end{proof}

With Lemma~\ref{lem:subgame1} and~\ref{lem:distinguish} we can have the following corollary:
\begin{corollary}
\label{lem:update}
    Rational wardens will honestly follow the update protocol.
\end{corollary}

\begin{proof}
    In the game-theoretic model, wardens could be bribed to sign two different states during a \vc update. However, since we require $2f+1$ accountable wardens to sign the state—wardens who can be penalized for such misbehavior—the security of both the update and unilateral closing processes can be analyzed using the same game model. This model discussed in Appendix~\ref{sec:subgame}, yields the same result for the update protocol. The proof for Lemma~\ref{lem:subgame1} also applies to the update protocol. In this context, $P_l$ represents honestly signing a single state with the latest sequence number, while $P_o$ indicates signing multiple states with the same sequence number. During the closing procedure, the closing party can still punish double-signing through fraud proofs, just as outdated state publishing is penalized. Therefore, for the update protocol, the SPNE strategy for wardens remains $P_l$, meaning they sign only one state. Additionlly, the proof for Lemma~\ref{lem:distinguish} also applies to the update protocol and implies that wardens cannot determine whether Ingrid is able to punish them for signing two different states.
\end{proof}

After obtaining the final SPNE utility of Subgame 1, we substitute the utility function with the subgame notation in the closing game initiated by Bob. The results are shown in Fig.~\ref{fig:bcloseimpE} and~\ref{fig:bclosepE}. Here, we subtract a small value, $\epsilon$, from the unilateral closing party's final utility to represent the cost incurred after unilateral closing. This value $\epsilon$ accounts for potential costs such as time and incentive fees. By incorporating $\epsilon$ into the utility function, we can also demonstrate the party's preference between collaborative and unilateral closing.

For the interaction between Bob and Ingrid, we present the following lemmas:

\begin{lemma}
\label{lem:imperfect}
    For the closing game $\mathcal{G}^{no\_knowledge}_B$, which is started by Bob when Ingrid has no knowledge about \vc, the joint SPNE strategy $\{\sigma_B(\bot)=Old, \sigma_I(Old)=\sigma_I(New)=Ignore\}$.
\end{lemma}

\begin{proof}
In the closing game $\mathcal{G}^{no_knowledge}_{B}$, Ingrid lacks prior knowledge about \vc and thus cannot distinguish whether the closing requests sent by Bob are outdated or current. In this scenario, the closing game initiated by Bob becomes an Extensive Form Game (EFG) with imperfect information. It is known that the SPNE of an EFG with imperfect information is equivalent to the Nash Equilibrium (NE) strategy of the Normal Form Game (NFG) transformed from the EFG~\cite{osborne1994course}. Therefore, to find the SPNE, we first transform the closing game into an NFG, as shown in Table~\ref{tab:NFG}.

    \begin{table}[h]\tiny

\resizebox{\columnwidth}{!}{%
\begin{tabular}{|c|c|c|c|}
\hline
      & $Ignore$                                             & $Agree$                    & $Disagree$                 \\ \hline
$Uni$ & $(\alpha-\epsilon,\alpha)$                           & $(\alpha-\epsilon,\alpha)$ & $(\alpha-\epsilon,\alpha)$ \\ \hline
$Old$ & \textcolor{blue}{$(\alpha-\epsilon,\alpha)$} & $(\alpha+d_B,\alpha-d_B)$  & $(\alpha,\alpha-\epsilon)$ \\ \hline
$New$ & $(\alpha-\epsilon,\alpha)$                           & $(\alpha,\alpha)$          & $(\alpha,\alpha-\epsilon)$  \\ \hline
\end{tabular}%
}
\caption{Transformed NFG of Bob's closing game with imperfect information}
\label{tab:NFG}
\end{table}

By comparing the utility functions in the table, we observe that Bob's strategy, $Old$, dominates all other strategies, while $Ignore$ is the dominant strategy for Ingrid. Therefore, the joint strategy $\sigma^* = \{\sigma_B, \sigma_I\} = \{Old, Ignore\}$ becomes the SPNE strategy for the game $\mathcal{G}^{no_knowledge}_{B}$. Two possible outcomes may occur after Ingrid ignores Bob's request. First, to close the \vc part as quickly as possible, Bob may initiate Protocol~\ref{alg:uni-AB} to close it unilaterally. According to Lemma~\ref{lem:subgame1}, Ingrid will not lose any coins, assuming Bob is a rational party. Second, Ingrid may gain knowledge about \vc either by receiving Alice's closing request or by closing the Alice-Ingrid part of \vc. If the Ingrid-Bob part has not been closed by this point, the closing game will shift to an EFG with perfect information, as depicted in Fig.~\ref{fig:bclosepE}. 
\end{proof}

\begin{lemma}
\label{lem:perfect}
    For the closing game $\mathcal{G}^{knowledge}_{B}$, which is started by Bob when Ingrid has knowledge about \vc, the joint SPNE strategy is $\{\sigma_B(\bot)=New, \sigma_I(New)\in \{Agree,Ignore\}\}$.

\end{lemma}

\begin{proof}
In the closing game $\mathcal{G}^{knowledge}_B$, Ingrid possesses knowledge about \vc, making the game an EFG with perfect information. As shown in Fig.~\ref{fig:bclosepE}, for the subgame $\mathcal{\varphi}((Old))$ where Bob takes the strategy $Old$, the SPNE strategy for Ingrid is $\sigma_I(Old) = \{Ignore\}$, as it is the only strategy that prevents Ingrid from incurring a loss. On the other hand, in the subgame $\mathcal{\varphi}((New))$ where Bob chooses the strategy $New$, Ingrid has two SPNE strategies, $\sigma_I(New) \in \{Agree, Ignore\}$. Since Bob can only maximize his utility when both parties agree, with $u((New, Agree)) = (\alpha, \alpha)$, the SPNE strategy for Bob in the game $\mathcal{G}^{knowledge}_B$ is $New$. 
\end{proof}

Additionally, we discuss a special situation where both Alice and Bob are offline, requiring Ingrid to run Protocol~\ref{alg:uni-AB} to unilaterally close both the Alice-Ingrid and Ingrid-Bob parts of \vc. According to Lemma~\ref{lem:subgame1}, Ingrid will not be cheated in this scenario. Furthermore, Ingrid is also unable to cheat Alice or Bob. The cross-checking scheme continues to ensure the security of \vc in a game-theoretic setting, as demonstrated in the following Lemma:

\begin{lemma}
\label{lem:rationalupdate}
    The cross-checking scheme guarantees rational main parties (e.g., Alice and Bob) will honestly follow the update protocol.
\end{lemma}

\begin{proof}
    According to Lemma~\ref{lem:secupdate}, both parts of \vc can only be closed using the same state, enforced by the cross-checking scheme. In this case, there is no incentive for two rational parties to collude and attempt to cheat the other party by deviating from the update protocol and trying to close both parts of \vc with different states.
\end{proof}

Consequently, with all the conclusions above, we can prove the game-theoretic security of \vc:

\Thmgamesec*

\begin{proof}
    
    To prove the game-theoretic security as defined in Definition~\ref{def:Game-theoretic security}, we need to show that, in the utility function of the SPNE strategy for the \vc closing game, each main party's utility is no less than $\alpha - \epsilon$.

    To start with, according to Lemma~\ref{lem:subgame1}, the SPNE strategy for each warden is $\sigma_W(\bot) = P_l$, which results in the utility for the unilateral closing party $P$ and their counterparty being $(\alpha, \alpha)$. Considering the cost of unilateral closing, the final SPNE utility of \emph{subgame 1} for the main parties is $(u_{closing}, u_{counter}) = (\alpha - \epsilon, \alpha)$. Applying this utility to games $\mathcal{G}^{no_knowledge}_B$ and $\mathcal{G}^{knowledge}_B$, we derive the final utility function for the \vc closing game. According to Lemma~\ref{lem:imperfect}, the SPNE strategy is the joint strategy $(\sigma_B,\sigma_I)=(Old, Ignore)$, which leads final utility $(u_B,u_I)=(\alpha-\epsilon,\alpha)$. According to Lemma~\ref{lem:perfect}, the SPNE strategy is the joint strategy $(\sigma_B,\sigma_I)=(New, \{Agree,Ignore\})$, which leads final utility $(u_B,u_I)\in \{(\alpha,\alpha),(\alpha-\epsilon,\alpha)\}$. Furthermore, Lemma~\ref{lem:rationalupdate} demonstrates that rational parties will not deviate from the update protocol or attempt to close both parts of \vc with different states. Therefore, the \vc closing protocol is game-theoretically secure.
\end{proof}

\begin{figure}
    \centering
    \includegraphics[width=\linewidth]{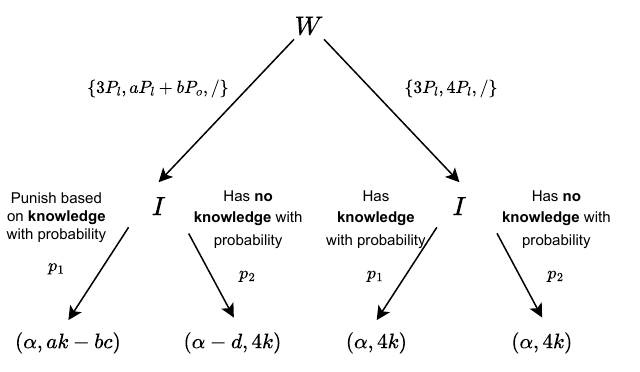}
    \caption{Unilateral closing game started by Ingrid. From the warden's perspective, Ingrid may have knowledge about \vc with possibility $p_1$, and have no knowledge with possibility $p_2$ }
    \label{fig:distinguish}
\end{figure}

\section{Security analysis for multi-hop \vc}\label{sec:multi-analysis}

\subsection{Security under Byzantine model}\label{subsec:multi-byz}

In this section, we show that multi-hop \vc is secure in the Byzantine model, i.e., satisfies balance security and liveness, similar to two-hop \vc. 

\begin{theorem}
\label{thm:multihopsecurity}
    Multi-hop \vc achieves \textbf{balance security} for honest parties under asynchrony, assuming at most $f$ Byzantine wardens in each committee.
\end{theorem}

\begin{proof}
    During the opening procedure of \vc, the initial state and warden identities will be correctly set up as long as there is at least one honest party among all four involved. Consequently, no Byzantine party can unilaterally close the multi-hop \vc under different conditions that have not been agreed upon by the honest party.
    
    For the optimistic closing, where all parties are responsive, assume Carol is the only honest party. In this case, Carol will only agree to update requests from Bob and Dave if they are correct and consistent. Therefore, the honest party will not lose any coins by updating the underlying TPC, even if both counterparties are Byzantine.
    
    The pessimistic closing remains secure, as demonstrated in the proof of Theorem~\ref{thm:security}, and also holds under Byzantine collusion. If Carol is the only honest party, the only way for Alice, Bob, and Dave to collude and attempt to cheat Carol is by closing the Bob-Carol and Carol-Dave parts differently through unilateral closing. However, the cross-checking scheme ensures that all parts of the multi-hop \vc can only be closed with the same state, thus preventing incorrect closings from occurring.
\end{proof}

\begin{theorem}
\label{thm:multihopliveness}
    Multi-hop \vc achieves \textbf{liveness} for honest parties under asynchrony, assuming at most $f$ Byzantine wardens in each committee.
\end{theorem}

\begin{proof}
    We prove liveness from two aspects: the liveness of the update request and the liveness of the close request. To hinder the execution of any valid update request in a multi-hop \vc, at least $f+1$ wardens from the update party's committee would need to ignore the signing request. However, our security model assumes that there are at most $f$ Byzantine wardens in each committee. Therefore, the liveness of update requests is guaranteed for multi-hop \vc.

    For the liveness of the closing, we consider two cases: First, if all parties are honest and responsive, they will agree on a correct collaborative closing request. In the second case, if at least one main party is malicious, the liveness is ensured by the pessimistic closing protocol. In this scenario, the closing request will eventually be committed, as there will always be wardens who publish the states they have stored on-chain. Thus, \vc can always be closed, even if some parties are offline or misbehave, ensuring that the liveness of closing is upheld.
\end{proof}

\subsection{Security under game-theoretic model}\label{subsec:multi-game}

We prove that our multi-hop \vc is still secure under the game-theoretic model by showing that the possible closing game in multi-hop \vc can be reduced to the two-hop \vc case. 

\begin{theorem}
    The multi-hop \vc is game-theoretic secure.
\end{theorem}

\begin{proof}
    To prove that the multi-hop \vc is secure, we address each phase separately. As in the two-hop case, the security of the update phase is merged with the unilateral closing phase.
    
    First, we show by contradiction that opening the multi-hop \vc remains game-theoretically secure. Consider a multi-hop \vc involving four parties: Alice, Bob, Carol, and Dave, as depicted in Fig.~\ref{fig:multihop}. Suppose Bob and Carol collude to open their channel with a different state than what was agreed upon with Alice and Dave. Note that Bob and Carol can only incorrectly update their own payment channel. Assume the agreed-upon state is $(a, b)$, but they instead lock the state as $(a, b' < b)$. If, in the final state, Dave must pay Alice $c$ coins, and $b' < c < b$, Bob would still be required to pay Alice $c$ coins in the underlying channel. If Carol cannot pay accordingly (due to locking fewer coins), Bob becomes the party who loses coins. Thus, a utility-maximizing Bob (or Carol, symmetrically) would not open their channels in inconsistent states. We conclude that rational parties will correctly open the multi-hop \vc.

    To show the security of the closing game, we first prove that the unilateral closing subgame is secure. The update protocol and the unilateral closing game remain the same as in the two-hop case, requiring $2f+1$ wardens from the local committee to publish the state. Additionally, wardens cannot determine whether intermediary parties (e.g., Bob and Carol) know the \vc state, as collusion between the intermediary and the main \vc parties (like Alice and Dave) is always possible -- similar to the two-hop \vc scenario. Therefore, the security of the unilateral closing subgame is inherited from the security analysis of the two-hop \vc presented earlier.

    To complete the reduction of the multi-hop \vc to the two-hop version, we must demonstrate that the entire closing game is secure. The key observation is that when one party, such as Dave, wishes to close the multi-hop \vc, only two participants are involved in the closing game: Dave and Carol, despite the involvement of several intermediaries. This is because Dave ultimately wants to modify the underlying payment channel with Carol, which can only be achieved with Carol's consent or with the help of the wardens. As a result, the presence of other parties does not affect the game tree or change the utility outcomes. Given this and the security of the unilateral closing subgame, we conclude that the closing game shown in Fig.~\ref{fig:bclose} and~\ref{fig:bclosep} also reflects the closing game of multi-hop \vc. Additionally, the cross-checking scheme remains effective for multi-hop \vc and ensures that all parts of the \vc are closed with the same state, preventing rational parties from colluding to close \vc with different states. Thus, the security of closing a multi-hop \vc is directly inherited from the security analysis in Section~\ref{sec:gameproof}. 
\end{proof}

\end{document}